\documentclass{scrartcl}
\pdfoutput=1

\usepackage[utf8]{inputenc}
\usepackage{graphicx}

\usepackage{url}
\usepackage{amsmath}
\usepackage{amsthm}
\usepackage[algoruled,linesnumbered,noend,algo2e]{algorithm2e}
\usepackage{enumerate}
\usepackage{xspace}
\usepackage{booktabs}
\usepackage{multirow}
\usepackage{array}
\usepackage{xcolor}
\usepackage{textcomp}
\usepackage{microtype}

\newcommand{\extlcp}{\texttt{extLCP}\xspace}
\newcommand{\beetl}{\texttt{BEETL}\xspace}
\newcommand{\beetlem}{\texttt{BEETL2}\xspace}
\newcommand{\bwtlcpem}{\texttt{ble}\xspace}
\newcommand{\gsais}{\texttt{gsa-is}\xspace}
\newcommand{\egsa}{\texttt{egsa}\xspace}
\newcommand{\egap}{\texttt{egap}\xspace}

\newtheorem{proposition}{Proposition}
\newtheorem{observation}{Observation}
\newtheorem{definition}{Definition}
\newtheorem{theorem}{Theorem}
\newtheorem{lemma}{Lemma}

\usepackage[hyperfootnotes=false]{hyperref}
\hypersetup{%
  colorlinks = true,
  linkcolor = red!60!black,
  urlcolor = red!50!black,
  citecolor = blue!60!black,
  pdfauthor = {Bonizzoni et al.},
  pdfkeywords = {},
  pdftitle = {Computing the multi-string BWT and LCP array in external memory},
  pdfsubject = {},
  pdfpagemode = UseNone
}

\begin{document}

\title{Computing the multi-string BWT and LCP array in external memory}

\author{
  Paola Bonizzoni$^*$ \and
  Gianluca Della Vedova \and
  Yuri Pirola \and
  Marco Previtali \and
  Raffaella Rizzi
}
\date{}
\publishers{\smaller
  DISCo, Università degli Studi di Milano--Bicocca,
  Milan, Italy\\[1ex]
  \smaller
  $^*$Corresponding author \href{mailto:paola.bonizzoni@unimib.it}{paola.bonizzoni@unimib.it}
}

\maketitle

\begin{abstract}
Indexing very large collections of strings, such as those produced by the widespread next
generation sequencing technologies, heavily relies on
multi-string generalization of the Burrows-Wheeler Transform (BWT):
large requirements of in-memory approaches have stimulated
recent developments on external memory algorithms.
The related problem of computing the  Longest
Common Prefix (LCP) array of a set of strings is instrumental to compute the suffix-prefix overlaps among strings, which is an
essential step for many genome assembly algorithms.
In a previous paper, we presented an in-memory divide-and-conquer method for
building the BWT  and LCP where we merge partial BWTs with a forward approach to sort suffixes.

In this paper, we propose  an alternative backward strategy  to  develop an external memory  method  to simultaneously build the BWT and the LCP array on a collection of $m$ strings of different lengths. The algorithm over a set of strings having constant length $k$ has \(\mathcal{O}(mkl)\) time and I/O volume, using \(\mathcal{O}(k + m)\) main memory,  where $l$ is the maximum value in the LCP array.
\end{abstract}


\section{Introduction}
In this paper we address the problem of constructing, simultaneously and in external memory, the Burrows-Wheeler Transform (BWT)
and the Longest Common Prefix (LCP) array for a large collection of strings.
The widespread use of Next-Generation Sequencing (NGS) technologies, that are
producing everyday several terabytes of data that  has to be analyzed, requires efficient
strategies to index very large collections of strings.
For example, common applications in metagenomics require  indexing of  collections of strings (reads)
that are sampled from several genomes: those strings can easily contain more than $10^{8}$ characters.
In fact, to start a catalogue of the human gut microbiome, more than 500GB of data have
been used~\cite{qin2010human}.

The Burrows-Wheeler Transform (BWT)~\cite{Burrows1994} is a reversible transformation of a
text that was originally designed for text compression; it is used for example in the
\texttt{bzip2} program.
The BWT of a text $T$ is a permutation of its symbols and is strictly related to the
Suffix Array of $T$.
In fact, the i-${th}$ symbol of the BWT is the symbol preceding the i-${th}$ smallest suffix of $T$ according to the lexicographical sorting of the suffixes of $T$.
The BWT has gained importance beyond its initial purpose, and has
become the basis for self-indexing structures such as the FM-index~\cite{Ferragina2005}, which
allows to efficiently perform important tasks such as searching a pattern in a text~\cite{Ferragina2005,Li15112014,Rosone2013}.
The generalization of the BWT (and the FM-index) to a collection of strings was
introduced in~\cite{MantaciCPM2005,MantaciTCS2007}.

An entire generation of recent Bioinformatics tools heavily rely on the notion of BWT.
Representing the reference genome with its FM-index is the basis of the most
widely used aligners, such as Bowtie~\cite{Langmead2009},
BWA~\cite{LiBWABioinformatics2009,LiBWABioinformatics2010} and
SOAP2~\cite{LiSOAP2Bioinformatics2009}.
Still, to attack some other fundamental Bioinformatics problems, such as genome assembly,
an all-against-all comparison among the input strings is needed, especially to
find all prefix-suffix matches (or overlaps) between reads in the context of the Overlap-Layout-Consensus (OLC) approach based on string graph~\cite{Myers2005}. This fact justifies
the search for extremely efficient algorithms to compute the BWT on a collection of strings~\cite{li_exploring_2012,Valimaki2010,Ferragina2012,bauer_lightweight_2013}.
For example, SGA (String Graph Assembler)~\cite{Simpson2012} is a de novo genome assembler that builds a string graph from the FM-index of the collection of input reads.
In a preliminary version of SGA~\cite{Simpson2010}, the authors  estimated, for human sequencing data at a 20x coverage, the need of
700Gbytes of RAM in order to build the suffix array, using the construction algorithm
in~\cite{nong2009linear}, and the
FM-index.
Another technical device that is used to tackle the genome assembly in the OLC approach is the Longest Common Prefix (LCP) array of a collection of strings, which is instrumental to compute the prefix-suffix matches in the collection.

The construction of the BWT and LCP array of a huge collection of strings is a
challenging task.
A simple approach is constructing the BWT from the Suffix Array, but it is prohibitive for massive datasets mostly due to the main memory requirements.
A first attempt to solve this problem~\cite{Siren2009}
partitions the input collection into batches, computes the BWT for each batch and then merges the results.

The huge amount of available biological data has stimulated the development  of
efficient external memory algorithms (called, BCR and BCRext) to construct the BWT of a collection of strings~\cite{Bauer2011}.
Similarly, a lightweight approach to the construction of the LCP array (called \extlcp) was investigated in~\cite{Cox2016}. %
With the ultimate goal of obtaining  an external memory genome assembler,
LSG~\cite{DBLP:journals/jcb/BonizzoniVPPR16,Bonizzoni2017} is based on BCRext and
contains an external memory approach to compute the string graph of a set of reads.
In that approach, external memory algorithms to compute the BWT and the LCP
array~\cite{bauer_lightweight_2013,DBLP:conf/wabi/BauerCRS12} are fundamental.

In this context, we are considering a model of computation where memory is split into two
parts: a finite random access memory, and an unlimited sequential access disk.
Essentially, this model is an extension of the standard RAM model where we also have a
sequential access disk.

In this paper we present a new lightweight (external memory) approach to compute the BWT
and the LCP array of a collection of strings of different lengths, which is alternative to
\extlcp~\cite{Cox2016} and other approaches~\cite{Bonizzoni2016ANL,Holt2014,DBLP:conf/spire/EgidiM17,DBLP:journals/corr/Manzini16,Louza2017d}.
The literature is rich of in-memory methods~\cite{karkkainen_faster_2016,kasai2001linear}, as well as
some external-memory algorithm on a single text~\cite{karkkainen_faster_2016}.
From a theoretical point of view, we can transform a set of strings into an instance
consisting of a single text by concatenating the input strings after adding a distinct
sentinel for each string.
Anyway this would increase the alphabet size from \(\sigma\) to
\(m + \sigma\), and this effect must be taken into account.

The algorithm BCRext is proposed together with BCR and both are designed to work on huge
collections of strings (the experimental  analysis is on hundreds of millions of 100-long strings).
Especially \extlcp is lightweight because, on a collection of $m$ strings of length $k$,
it uses only  \(\mathcal{O}(m + \sigma^2)\) RAM space and essentially
\(\mathcal{O}(mk^{2})\) CPU time, with matching I/O volume,
under the usual assumption that the word size is sufficiently large to store all
addresses.

An important question is how to achieve the optimal  $\mathcal{O}(km)$ I/O volume.
BCRext~\cite{Bauer2011} incrementally computes the BWT of the collection $S$ via $k+1$ iterations.  At each iteration $l$ ($0\le l\le k$) the algorithm computes a partial BWT
$bwt_{l}(S)$ that is the BWT for the ordered collection of suffixes with length at
most  $l$.
This approach requires that, at each iteration $l$,  the symbols preceding the
suffixes of $S$ with length $l-1$ must be inserted at their correct positions into $bwt_{l-1}(S)$, that is each  $l$  iteration simulates the insertion of the suffixes with length $l$ in the ordered collection of the suffixes with length at most $l-1$.  Updating  the partial BWT $bwt_{l}(S)$ in external memory, the process  requires a sequential scan of the file containing the  information of the partial $bwt_{l -1}(S)$. Thus the I/O volume at each  iteration $l$  is at least $m (l-1) \lg \sigma$ (since there are  $m$ suffixes for each length $l$ between $1$ to $l -1$).
Consequently the total I/O volume for computing $bwt_{k}(S)$ is  at least $O(m k^2)$.
More precisely, the BCRext algorithm in \cite{Bauer2011} that uses less RAM,  requires at each $l$ iteration an additional  I/O volume  given by $m  \lg (km)$, due to a process of ordering special arrays used to save RAM space.
Our algorithm for building the BWT and the LCP, differently from~\cite{Bauer2011},
consists of two distinct phases: the first phase  that has $O(m k)$ I/O volume  and time complexity produces $k+1$ arrays $B_0, \ldots, B_k$, each array $B_l$ lists the symbols preceding the suffixes with length exactly $l$ according to the lexicographical ordering of such suffixes.
The second phase computes the interleave of vectors $B_0, \cdots, B_k$ that is equal to the BWT $B$ of $S$. Indeed, the BWT $B$ is an interleave of the arrays
$B_0, \ldots, B_k$, since the ordering of symbols in $B_{l}$ is preserved in the final
BWT $B$, i.e., $B$ is \emph{stable} w.r.t.~each array $B_0, \ldots, B_k$.
Inspired by~\cite{Holt2014}, we perform this step by a number of $L$ iterations, where $L$
the length of the longest substring that has at least two occurrences in $S$.
Thus the merging operation takes fewer iterations than BCRext (the latter requires $k$ iterations).
Observe that at each iteration of the merging procedure of the arrays $B_0, \ldots B_{k}$, a partial LCP array is computed  to get the  final LCP array at the last iteration.

Our algorithm has \(\mathcal{O}(mkl\sigma)\) time complexity, uses \(\mathcal{O}(mkl\max\{\lg m, \lg l\})\)
I/O volume, and \(\mathcal{O}(\sigma w + k w + m \lg \sigma + \lg l)\) main memory, where $l$ is the maximum value in the
LCP array and \(w\) is the space required to store a memory address.
Moreover, our approach is entirely based on linear scans (i.e., it does not contain a
sorting step) which makes it more amenable to actual disk-based implementations.
We point out that $l\le k$, therefore our time and I/O complexities are at least as good
as those of \extlcp~\cite{Cox2016} when building the data structures for massive sets of short sequences over a constant alphabet and if \(\lg m\) and \(\lg l\) are smaller than the word size (which is usually the case).
The RAM usage of our approach and that of \extlcp are not easily comparable, since they are respectively $\mathcal{O}(\sigma w + k w + m \lg \sigma + \lg l)$ and
$\mathcal{O}(m+\sigma^{2})$.
If we suppose that we can store a memory address in a memory word, our RAM usage is \(\mathcal{O}(\sigma + k + m\lg \sigma + \lg l)\).
This means that, in theory, when building the BWT and the LCP of few large strings \extlcp will use less RAM than the method presented in this paper.
We point out that our algorithm works also on a set of reads having different lengths, and the following sections describe the algorithm referring to that case.

While writing our paper, two similar approaches have
appeared in the literature~\cite{DBLP:conf/spire/EgidiM17,Louza2017d}.
The method proposed in~\cite{DBLP:conf/spire/EgidiM17} starts from the BWT merging phase of~\cite{Holt2014}
to also build the LCP array using a small amount of memory.
We point out that our paper and~\cite{DBLP:conf/spire/EgidiM17} are two
independent works.
Moreover, our focus is on an external-memory approach, which is not explicitly
pursued in~\cite{DBLP:conf/spire/EgidiM17}.
An extension to fully external memory computation of BWT and LCP of~\cite{DBLP:conf/spire/EgidiM17} is \egap~\cite{DBLP:journals/almob/EgidiLMT19}.
This method computes the data structures in three separate steps, (1) splitting the input sequences in subcollections such that the BWT can be computed in-memory, and (2-3) then merging them building the LCP along the way.
The I/O and time complexities of \egap are both \(\mathcal{O}(mkl)\), matching the ones of our algorithm when the alphabet is constant.

The method proposed in~\cite{Louza2017d} (\egsa) is a two-phase algorithm for the construction of the Generalized Enhanced Suffix Array, the LCP, and, optionally, the BWT.
In the first step the required data structures are build for each sequence in input, whereas in the second step the output is produced by merging the data structures built previously.
Although \egap can be seen as an evolution of \egsa, we included both tools in our experimental evaluation to highlight that building the Generalized Enhanced Suffix Array requires way more resources than directly building the BWT and the LCP.

The paper is laid out as follows.
In Section~\ref{sec:definitions} we provide the basic definitions we will use in the following.
In Section~\ref{sec:sketch-algorithm} we give a high level description of the method proposed in this paper and illustrate the backward and forward strategies to merge partial arrays.
In Section~\ref{sec:partition-suffixes} and Section~\ref{sec:merge-suffixes} we dive into the details of our algorithm.
In Section~\ref{sec:complexity} we analyze the time and I/O complexities of our method.
In Section~\ref{sec:exp-eval} we provide an experimental analysis of our tool and a comparison with other tools available in the literature.
Finally, in Section~\ref{sec:conclusions} we recap the contributions of this paper.

\section{Preliminaries}
\label{sec:definitions}

Let $\Sigma = \{c_0, c_1, \ldots, c_{\sigma}\}$ be a finite alphabet where $c_0 = \$$
(called \emph{sentinel}), and $c_0 < c_1 < \cdots < c_{\sigma}$ where $<$ specifies the lexicographic ordering over alphabet $\Sigma$.
We consider a collection $S=\{s_1, s_2, \ldots, s_m\}$ of $m$ strings (reads), where each string $s_j$ consists of $k_j$ %
symbols  over the alphabet $\Sigma \setminus \{\$ \}$ and is terminated by the symbol \$, such that $k_j+1$ is the total length of $s_j$. The set $S$ is intended as a sequence of strings, where $s_j$ is the j-$th$ string in the set.
The $i$-{th} symbol of string $s_j$ is denoted by $s_j[i]$,  and the substring
$s_j[i]s_j[i+1] \cdots  s_j[t]$ of $s_j$ is denoted by $s_j[i:t]$. We will refer to $s_j[k_j]$ as the \emph{last character} of the string $s_j$ and is the character immediately before the sentinel.
The \emph{suffix} and \emph{prefix} of $s_j$ with length $l$ are the
substrings %
$s_j[k_j-l+1: k_j+1]$ (denoted by $s_j[k_j-l +1:]$)
and $s_j[1: l]$ (denoted by
$s_j[:l]$) respectively. Observe that $l$ counts, for the suffix, only the characters which are in $\Sigma \setminus \{\$ \}$ (excluding \$).
Then, the suffix and prefix with length $l$ of a string $s_j$ will be called the \emph{$l$-suffix} and \emph{$l$-prefix} of $s_j$, respectively. In particular, the $0$-suffix is the suffix uniquely composed of the sentinel \$.
In the following we will denote with $K$ the total length of the input reads (including the sentinel \$).

Given the lexicographic ordering $X$ of the suffixes of $S$, %
the \emph{Suffix Array} is the $K$-long %
array $SA$ where the element $SA[i]$ is equal to $(p, j)$ if and only if the $i$-{th} element of $X$ is the  $p$-suffix of string $s_{j}$. %
We make the assumption that a suffix $s\$$ from string $s_i$ is lexicographically smaller than the identical suffix $s\$$ from a different string $s_j$ if $i < j$. In other words, the two identical suffixes are ordered accordingly to the order of their origin strings in $S$. This assumption guarantees the uniqueness of the Suffix Array for the collection $S$.

The definition of suffix array we provide is slightly different than the one conventionally used, where \(SA[i] = (p, j)\) refers to the suffix of the string \(s_j\) starting at position \(p\).
We have decided to abide by this definition to ease the presentation of the method in the following sections.

The \emph{Burrows-Wheeler Transform (BWT)} of $S$ is the $K$-long %
array $B$ where if $SA[i] = (p,j)$, then $B[i]$ is the first symbol of the $(p+1)$-suffix of $s_j$ if $p < k_j$, otherwise $B[i]= \$$.
In other words $B$ consists of the symbols preceding the ordered suffixes of $X$, where the
preceding symbol is the sentinel \$ when the suffix is the complete string $s_j$ (i.e., the $k_j$-suffix).

The \emph{Longest Common Prefix (LCP) array} of  $S$ is the $K$-long %
array $LCP$  such that $LCP[i]$  is the length of the longest prefix shared by suffixes $X[i-1]$ and $X[i]$.  Conventionally, $LCP[1]=-1$.

Now, we can give the definition of \emph{Interleave} of a generic set of arrays, that will be used extensively in the following.

\begin{definition}
	\label{def:interleave}
	Given $n+1$ arrays $V_0, V_1, \ldots, V_n$, then an array $W$ is an \emph{interleave} of $V_0, V_1, \ldots, V_n$ if
	$W$ is the result of merging the arrays such that: (i) there is a 1-to-1 function $\psi_W$
	from the set $\cup_{i=0}^n \{ (i,j): 1\le j\le |V_i|\}$ to the set $\{ q : 1\le q \le |W|
	= \sum_{i} |V_{i}|\}$, (ii)
	$V_i[j] = W[\psi_W(i,j)]$ for each $i,j$, and (iii) $\psi_W(i, j_1) < \psi_W(i, j_2)$ for each $j_1 < j_2$.
\end{definition}

The interleave $W$ is an array giving a fusion of $V_0, V_1 \ldots, V_n$ which preserves
the relative order of the elements in each one of the arrays. As a consequence, for each
$i$ with $0 \leq i \leq n$, the  j-${th}$ element of $V_i$ corresponds to the j-${th}$
occurrence in $W$ of an element of $V_i$. This fact allows to encode the function $\psi_W$ as an array $I_W$ such that $I_W[q] = i$ if and only if $W[q]$ is an element of
$V_i$. By observing that $W[q]$ is equal to $V_{I_W[q]}[r]$ where $r$ is the number of integers equal to $I_{W}[q]$ in $I_W[:q]$,  %
it is easy to show how to reconstruct $W$ from $I_W$
(see Algorithm~\ref{alg:interleave-encoding} where the array $pos$ keeps, for each index $i$ from $0$ to $n$, such number $r$ while scanning array $I_W$). %

In the following, we will refer to array $I_W$ as \emph{interleave-encoding} (or simply \emph{encoding}). Figure~\ref{fig:interleave} shows an example of an interleave of four arrays and its encoding.

\begin{figure}[t]
	\centering\renewcommand{\baselinestretch}{1.15}\normalsize
	\begin{minipage}{0.3\textwidth}%
		\begin{tabular}{|c|c|c|c|}
			\hline
			$V_{0}$&$V_{1}$&$V_{2}$&$V_{3}$\\\hline
			T&C&A&A\\
			T&G&C&A\\
			A&&C&T\\
			G&&&T\\
			C&&&\\
			\hline
		\end{tabular}
	\end{minipage}%
	\quad%
	\begin{minipage}{0.3\textwidth}%
		\begin{tabular}{|c|c|}
			\hline
			$W$&$I_{W}$\\\hline
			T&0\\
			A&2\\
			A&3\\
			A&3\\
			C&1\\
			C&2\\
			C&2\\
			G&1\\
			T&0\\
			A&0\\
			T&3\\
			G&0\\
			T&3\\
			C&0\\
			\hline
		\end{tabular}
	\end{minipage}
	\caption{Example of an interleave $W$ of four arrays $V_0, V_1, V_2, V_3$.}
	\label{fig:interleave}
\end{figure}

\begin{algorithm2e}[htb!]
	\For{$i\gets 0$ to $n$}{%
		$pos[i] \gets 0$\;
	}

	\For{$q \gets 1$ to $|I_W|$}{%
		$i \gets I_W[q]$\;
		$pos[i] \gets pos[i] +1$\;
		$W[q] \gets V_i[pos[i]]$\;\label{alg:interleave-encoding:element}
	}

	\caption{Reconstruct the interleave $W$ from the encoding $I_W$}
	\label{alg:interleave-encoding}
\end{algorithm2e}

\section{The algorithm}
\label{sec:sketch-algorithm}

In this section we will provide a sketch of our algorithm. %
Let $k$ be the maximum length of a string in $S$ (excluding \$) and let $X_l$ and $B_l$ ($0 \le l \le k$) be arrays of length at most $m$ such that $X_l[i]$ is the
i-${th}$ smallest $l$-suffix among all the $l$-suffixes of the strings of $S$ and $B_l[i]$ is the symbol preceding %
$X_l[i]$.
In particular, $X_0$ and $B_0$ list respectively the $0$-suffixes and the last characters of the input strings in their order in the set $S$.
Observe that $B_l$ is a subsequence of the BWT $B$ of $S$,
and it is easy to see that $B$ is an interleave of the $k+1$ arrays $B_0, B_1, \ldots, B_k$, since the ordering of symbols in $B_{l}$ ($0 \leq l \leq k$) is preserved in $B$. %

Similarly, the lexicographic ordering $X$ of all suffixes of $S$ is an interleave of the arrays $X_0, X_1, \ldots, X_k$.
Let $I_B$ be the encoding of the interleave of arrays $B_0, B_1, \ldots, B_k$ giving the BWT $B$, and let $I_X$ be the encoding of the interleave of arrays $X_0, X_1, \ldots, X_k$ giving $X$. Then it is possible to show that $I_B = I_X$.
Now, given $I_B$ it is immediate to reconstruct $B$ by using Algorithm~\ref{alg:interleave-encoding}.

In the following, we will call $B_0, B_1, \ldots, B_k$ and  $X_0, X_1, \ldots, X_k$ as \emph{partial BWTs} and \emph{partial Suffix Arrays}, respectively. Figure~\ref{fig:example-vectors-B} shows an example of partial BWTs and partial Suffix Arrays for a set of $m=3$ reads on alphabet $\{A,C,G,T\}$, %
whose interleaves $B$ and $X$ (BWT and sorted suffixes, respectively) and the encoding $I_B=I_X$ are reported in the first, second and third columns of Figure~\ref{fig:example-interleave}.

\begin{figure}[t]
	\centering\renewcommand{\baselinestretch}{1.15}\normalsize%
	\begin{minipage}{0.79\textwidth}%
		$s_{1}$: \verb+TCGT$+\hspace{2em}
		$s_{2}$: \verb+CT$+\hspace{2em}
		$s_{3}$: \verb+ACA$+\\
	\end{minipage}\\[.5em]
	\begin{minipage}{0.2\textwidth}%
		\begin{tabular}{|cc|}
			\hline
			$B_{0}$&$X_{0}$\\\hline
			T&\$\\
			T&\$\\
			A&\$\\
			\hline
		\end{tabular}
	\end{minipage}%
	\begin{minipage}{0.2\textwidth}%
		\begin{tabular}{|cc|}
			\hline
			$B_{1}$&$X_{1}$\\\hline
			C&A\$\\
			G&T\$\\
			C&T\$\\
			\hline
		\end{tabular}
	\end{minipage}%
	\begin{minipage}{0.2\textwidth}%
		\begin{tabular}{|cc|}
			\hline
			$B_{2}$&$X_{2}$\\\hline
			A&CA\$\\
			\$&CT\$\\
			C&GT\$\\
			\hline
		\end{tabular}
	\end{minipage}%
	\begin{minipage}{0.2\textwidth}%
		\begin{tabular}{|cc|}
			\hline
			$B_{3}$&$X_{3}$\\\hline
			\$&ACA\$\\
			T&CGT\$\\
			\hline
		\end{tabular}
	\end{minipage}%
	\begin{minipage}{0.2\textwidth}%
		\begin{tabular}{|cc|}
			\hline
			$B_{4}$&$X_{4}$\\\hline
			\$&TCGT\$\\
			\hline
		\end{tabular}
	\end{minipage}%

	\caption{An example of $m=3$ reads $s_1, s_2, s_3$ with maximum length $k=4$, together with the partial BWTs $B_{0}, B_{1}, B_{2}, B_{3}, B_{4}$ and
		the partial Suffix Arrays $X_{0}, X_{1}, X_{2}, X_{3}, X_{4}$.}
	\label{fig:example-vectors-B}
\end{figure}

\begin{figure}[t]
	\centering\renewcommand{\baselinestretch}{1.15}\normalsize%
	\begin{tabular}{|clcc|cc|}
		\hline
		$B$&$X$&$I_{B}=I_X$&$\mathit{LCP}$&&\\\hline
		T&\$&0&-1&$B_0[1]$&$X_0[1]$\\
		T&\$&0&0&$B_0[2]$&$X_0[2]$\\
		A&\$&0&0&$B_0[3]$&$X_0[3]$\\
		C&A\$&1&0&$B_1[1]$&$X_1[1]$\\
		\$&ACA\$&3&1&$B_3[1]$&$X_3[1]$\\
		A&CA\$&2&0&$B_2[1]$&$X_2[1]$\\
		T&CGT\$&3&1&$B_3[2]$&$X_3[2]$\\
		\$&CT\$&2&1&$B_2[2]$&$X_2[2]$\\
		C&GT\$&2&0&$B_2[3]$&$X_2[3]$\\
		G&T\$&1&0&$B_1[2]$&$X_1[2]$\\
		C&T\$&1&1&$B_1[3]$&$X_1[3]$\\
		\$&TCGT\$&4&1&$B_4[1]$&$X_4[1]$\\
		\hline
	\end{tabular}
	\caption{BWT $B$, sorted suffixes $X$, encoding $I_B=I_X$ and LCP array for the set of reads presented in Figure~\ref{fig:example-vectors-B}. The last two columns report, for each element of $B$ and $X$, its origin in arrays $B_l$ and $X_l$ (respectively).}
	\label{fig:example-interleave}
\end{figure}

Our algorithm for building the BWT $B$ and the LCP array %
consists of two distinct phases: in the first phase it computes each partial BWT $B_l$ ($0 \leq l \leq k$) by implicitly sorting the $l$-suffixes of $S$
(see Section~\ref{sec:partition-suffixes}), while in the second phase %
it determines %
$I_X=I_B$ (see Section~\ref{sec:merge-suffixes}) by a merging algorithm inspired by~\cite{Holt2014} (for merging two BWTs), thus allowing to reconstruct $B$ as an interleave of $B_0, \ldots, B_k$.
We slightly modified the approach in~\cite{Holt2014} in order to merge the arrays %
$B_0, \ldots, B_k$ into the BWT $B$ by implicitly merging the array $X_0, \ldots, X_k$ into the array $X$ (giving the lexicographic ordering of all suffixes of $S$).
The second phase computes, together with the BWT $B$, also the LCP array of the input set $S$.

We note that the definition of partial BWTs and the method sketched here hint to some relation between the partial BWTs and the positional BWT (pBWT) presented in~\cite{durbin2014efficient}, although the latter is presented for an alphabet of size \(2\).
Indeed, given a set of sequences, both reorder the characters at a given distance from one end of each sequence in input.
More precisely, each partial BWT is an ordering of all the elements of the sequences at a given distance from the end of them, whereas each column of the pBWT is an ordering of all the elements at a given distance from the start of the sequences.
In light of this fact, we can describe the two steps sketched in this section as follows: (i) build the pBWT of the input sequences reversed, and (ii) build the BWT and the LCP array by merging the columns of the pBWT.
Although we will not describe our method in terms of pBWT in the following
sections, we think that the connection we just highlighted further confirms
strong relations between multiple BWT-like data structures presented thorough
the years to index different structures (\textit{e.g.},
trees~\cite{ferragina2009compressing}, de Bruijn
graphs~\cite{bowe2012succinct,boucher2015variable,belazzougui2016bidirectional},
and circular patterns~\cite{hon2011succinct}), as recently shown
in~\cite{gagie2017wheeler}.

Both phases of our method apply a Radix Sort strategy to reorder the suffixes (i.e., the $l$-suffixes of $S$ in order to compute the partial BWT $B_l$ in the first phase, and the overall set of suffixes of $S$ in the second phase in order to compute $I_B$). The first phase iteratively computes the partial BWTs $B_0, B_1, \ldots, B_k$. Each iteration $l$ ($0 \leq l \leq k$) computes $B_l$ from the order of the $l$-suffixes (array $X_l$) implicitly computed by the previous iteration $l-1$ (array $X_{l-1}$). We point out that this algorithm adopts a LSD Radix Sort strategy that can be interpreted as ``global'', since suffixes are sorted from the rightmost to the leftmost character (that is, it adopts a LSD strategy), and the order of $X_l$ is implicitly obtained from the order of $X_{l-1}$ without applying the radix sort to each one of the sets of $l$-suffixes.

The second phase applies a MSD Radix Sort strategy since it reorders the suffixes from the leftmost to the rightmost characters, and can be performed in two different ways as described in the following section.

 \subsection{Backward and forward strategies for merging the partial BWTs}
 \label{sec:sketch-second-phase}

The encoding $I_X$ is basically computed by an iterative procedure starting by
the \emph{trivial} sorting given by taking first the suffixes of $X_0$, followed
by the suffixes of $X_1$, followed by the suffixes of $X_2$, etc., followed by
the suffixes of $X_k$ (\emph{trivial} interleave). Note that the encoding of the
\emph{trivial} interleave is given by $k$ runs of the integers from $0$ to $k$:
that is, $|X_0|$ integers equal to $0$, followed by $|X_1|$ integers equal to $1$, etc., followed by $|X_k|$ integers equal to $k$.
Starting from that sorting, the procedure applies a MSD Radix Sort strategy to sort the suffixes of $S$, by the first (leftmost) characters at the first iteration, then by the first two characters at the second iteration, etc., and finally by the first $k$ characters ($k$ is the maximum length of the strings in the input set $S$) at the k-$th$ iteration. More precisely, at the p-$th$ iteration, it computes the encoding of the interleave, giving the sorting by the first $p$ characters, from the interleave giving the sorting by the first $p-1$ characters (computed at the previous iteration). At the k-$th$ iteration the computed encoding is clearly $I_X$.

In the following, the interleave of arrays $X_0, X_1, \ldots, X_k$, giving the sorting of the suffixes by the first $p$ characters will be called \emph{$p$-interleave} and denoted as $X^p$, and its encoding will be denoted as $I_{X^p}$. The encoding $I_X$ is clearly equal to $I_{X^k}$ (and $X$ is equal to $X^k$).
We point out that $X^p$ is the list of all the suffixes in the input collection $S$ sorted by their prefixes of length $p$. In other words, $X^p$ includes also suffixes shorter than $p$. In this ordering, a suffix $s\$$ shorter than $p$ will come before any suffix having string $s$ as a prefix, and moreover such suffix will have the same position in all the orderings $X^q$ such that $q > p$.

Iteration $p$ computes the encoding $I_{X^p}$ from the encoding $I_{X^{p-1}}$ obtained at the iteration $p-1$. The first iteration $p=1$ computes $I_{X^1}$ from $I_{X^0}$ which is the \emph{trivial} encoding composed of $k$ runs of the integers from $0$ to $k$ ($I_{X^0}$ is the encoding of the $0$-interleave giving trivially the suffixes sorted by the first $0$ characters).

Two different strategies can be used for computing $I_{X^p}$ from $I_{X^{p-1}}$, which are based on the two following observations.

\begin{observation}
\label{obs:backward}
If $X^{p-1}[i]$ with length $l=I_{X^{p-1}}[i]$ is the $r$-th suffix preceded by a symbol $c$, then the suffix $c X^{p-1}[i]$ with length $l+1$ will be the $r$-th suffix in $X^p$ starting with $c$. Therefore, $I_{X^p}[j]$ will be equal to $l+1$, such that $j = s + r$, and $s$ is the number of symbols preceding suffixes of $X^{p-1}$ which are smaller than $c$. Observe that, when $c=\$$, then $c X^{p-1}[i]$ is actually the empty suffix having length $0$, and $s$ is equal to $0$.
\end{observation}

\begin{observation}
\label{obs:forward}
Let $[b,e]$ be the interval of positions related to the suffixes of $X^{p-1}$ sharing the first $p-1$ characters, and (among them) let us consider the $r$-th suffix having a given $c$ at position $p$. Then, such suffix will be at position $j = b+s+r,$ of $X^p$ ($j \in [b,e]$), where $s$ is the number of suffixes in the interval $[b,e]$ having a symbol smaller than $c$ at position $p$. Therefore, $I_{X^p}[j]$ will be equal to $l$.
\end{observation}

The first strategy (\emph{backward}) is based on Observation~\ref{obs:backward} and consists in scanning the encoding $I_{X^{p-1}}$. $|\Sigma|$ empty buckets are initialized (one for each alphabet symbol), and for each length $l=I_{X^{p-1}}[i]$, the symbol $c$ preceding the related suffix is obtained from the partial BWTs ($c$ is indeed the $i$-th symbol of the interleave of the partial BWTs encoded by $I_{X^{p-1}}$,  and will be the $t$-th of vector $B_l$ if suffix $X^{p-1}[i]$ is the $t$-th suffix in $X^{p-1}$ having length $l$). At this point, the length $l+1$ is added to the bucket related to symbol $c$, if $c$ is not the sentinel \$, otherwise the value $0$ is added (to the \$ bucket).
At the end of the iterations, the concatenation of the buckets following the lexicographical order of the symbols, provides the encoding $I_{X^{p}}$.

The second strategy (\emph{forward}) is based on Observation~\ref{obs:forward} and maintains a partitioning of the generic encoding $I_{X^p}$ into contiguous segments, which are called $p$-segments. A $p$-segment is an interval of positions which are related to suffixes sharing the first $p$ characters.
The forward strategy consists in scanning the $p-1$-segments of the encoding $I_{X^{p-1}}$ one after the other and uses $|\Sigma|$ initially empty buckets. For each $p-1$-segment $[b,e]$, its $(e-b+1)$ suffixes are considered, and each suffix length $l$ is added to a
bucket depending on the symbol at position $p$ in the suffix.
At the end of the iterations over the $p-1$-segment, the concatenation of the buckets following the lexicographical order of the symbols, provides the encoding $I_{X^{p}}$ between positions $b$ and $e$.

Both algorithms compute the interleave $I_X$ giving the Suffix Array as defined in Section 2 whose uniqueness is guaranteed by the radix sort strategy.

Figure~\ref{fig:merging-strategies} shows an example for the two strategies applied to a set of three reads.

In Section~\ref{sec:merge-suffixes} the backward strategy will be detailed alongside the computation of the LCP array. We refer to~\cite{jcb-bonizzoni-2019,DBLP:conf/cie/BonizzoniVNPPR18} for the details about the forward strategy.

\begin{figure}[h!t]
	\includegraphics[width=0.99\linewidth]{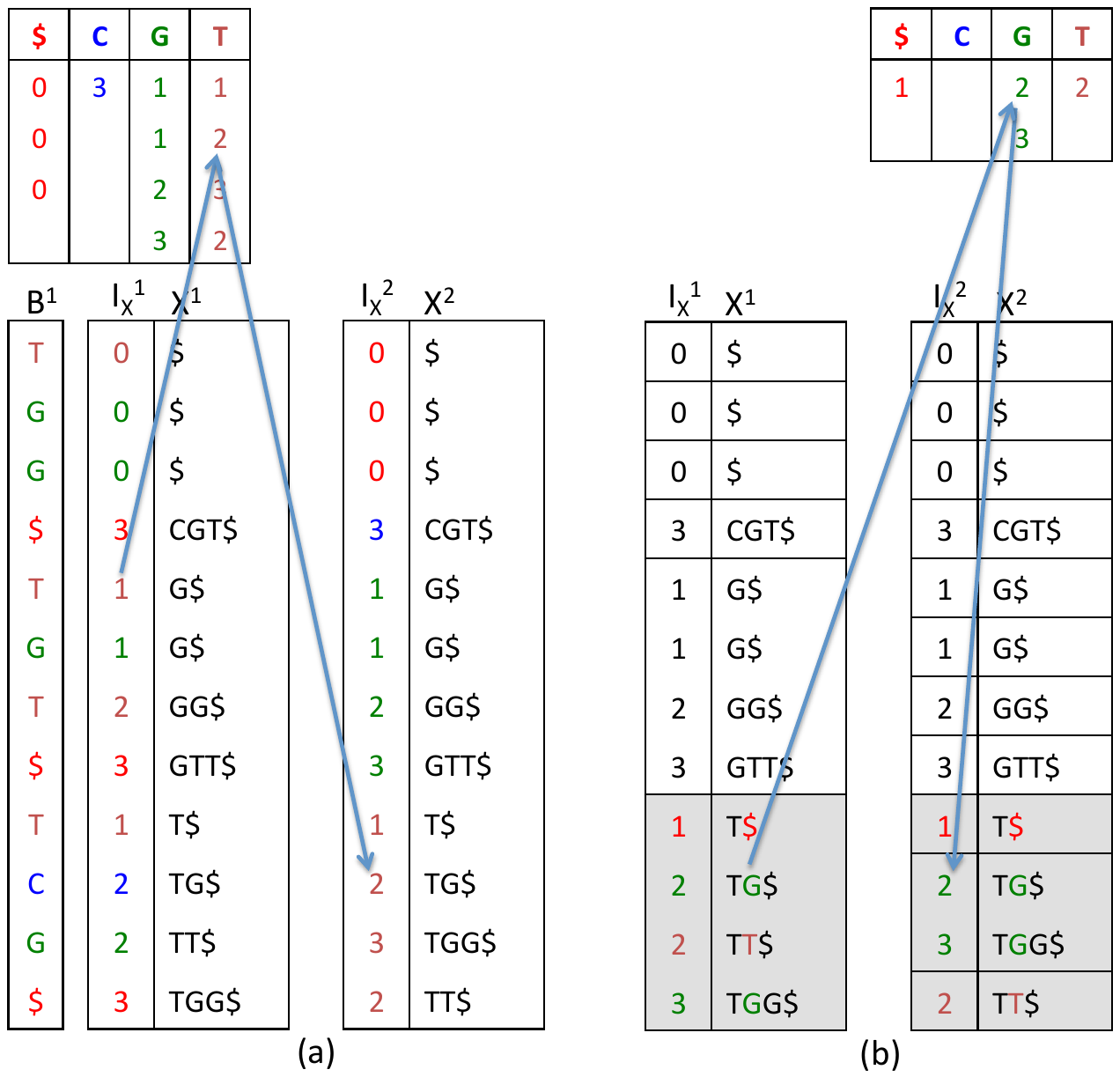}
	\caption{Example of backward and forward strategy. An example of application of backward (a) and forward (b) strategy for a set of three reads $s_1 = GTT$, $s_2 = CTG$ and $s_3 = TGG$. The encoding $I_{X^2}$ of the ordering by the first $2$ characters of the suffixes is obtained from the encoding $I_{X^1}$ of the ordering by the first character. Buckets \$,$C$,$G$,$T$ are depicted at the top of the figure.
	For (a) the interleave (tagged as $B^1$) of the partial BWTs encoded by $I_{X^1}$ is also shown (that is, the interleave of the preceding symbols of the suffixes), whereas for (b) the segments (on $I_{X^1}$ and $I_{X^2}$) are separated by horizontal bars; (b) shows how to obtain the portion of $I_{X^2}$ corresponding to the last $1$-segment of $I_{X^1}$ (this procedure must be obviously repeated for the previous segments), which is highlighted in grey and covers the last four positions of $I_{X^1}$, and refer to the suffixes sharing the first symbol $T$.
	Encoding values and symbols are colored in order to highlight the two strategies.}
	\label{fig:merging-strategies}
\end{figure}

\section{Computing the partial BWTs}
\label{sec:partition-suffixes}
The first phase of the method computes the partial BWTs $B_0, \ldots, B_k$ by first preprocessing the input strings $s_{1}, \ldots , s_{m}$ in order to obtain $k+1$ arrays $T_0,\ldots,T_{k}$ with length $m$,
where $T_l$ lists the characters %
such that $T_l[i]=s_i[|s_i|-l]$ %
when $0 \leq l \leq |s_i|-1$, %
$T_l[i]=\$$ when $l=|s_i|$, and $T_l[i] = \#$ when $l>|s_i|$ (where $\#$ is a padding symbol not belonging to the alphabet of the input strings). %
Section (a) of Figure~\ref{fig:example-iteration-partialbwts} reports an example of arrays $T_l$ for the three strings of Figure~\ref{fig:example-vectors-B}. Observe that $T_0$ lists the last characters $\langle s_1[k_1], s_2[k_2], \ldots, s_m[k_m] \rangle$ of the input strings in the same order the strings have in the set $S$, and $T_0$ is clearly equal to $B_0$. Observe that $T_l[i]$ (when different from $\#$) is the symbol preceding the $l$-suffix of string $s_i$.

The preprocessing step is a trivial task that iterates
over the input strings and outputs the $k+1$ arrays $T_0, \ldots, T_k$.
We can summarize this procedure as a loop that iterates over the input strings and performs the following steps.
Let $s_i$ be the input string and suppose that we already preprocessed the previous $i-1$ sequences.
We first reverse $s_i$ and produce the string $r_i$.
Then, for each position $l$ of $r_i$, we write the $l$-th character of $r_i$ at position $i$ of array $T_l$, padding this array with \# if it includes less than $i-1$ elements.
Finally, we write $\$$ at position $i$ of array $T_{|s_i|}
$ (padding it with \# if required) and move to the next sequence.

The partial BWTs $B_0, \ldots, B_k$ are computed by Algorithm~\ref{alg:sort-suffixes} which receives in input the arrays $T_0, \ldots, T_{k}$, and performs $k+1$ iterations.
Iteration $l$ ($0 \leq l \leq k$) computes $B_l$ from $X_l$ which is implicitly known and implicitly determines $X_{l+1}$ to be used in the next iteration.
More in detail, the ordering of array $X_l$ is known by means of a array $N_{l}$, with length at most $m$, such that $N_{l}[i]=q$ if and only if the i-${th}$ element of $X_{l}$ is the $l$-suffix from the input string $s_q$. In other words, position $i$ of array $N_{l}$ gives the index $q$ of the string whose $l$-suffix is the i-$th$ in $X_{l}$.
The partial BWT $B_l$ can be computed (see cycle at line~\ref{alg:sort-suffixes:compute-projections-start}) by scanning $N_l$, since $B_l[i]$ is (by definition) the symbol preceding the $l$-suffix of the string $s_q$, where the index $q$ is equal to $N_l[i]$, and can be obtained by accessing array $T_l$. Indeed, $B_l[i]=T_l[q]$. Observe that $B_l$ is treated by Algorithm~\ref{alg:sort-suffixes} as a list initially empty, and the symbol $c$ is appended to $B_l$ only if it is not the padding symbol \# (signaling that the originating string is shorter than $k$).
Note that, at the first iteration $l=0$, $N_0$ (which is set in cycle at line~\ref{alg:sort-suffixes:compute-N0})
is the sequence of indices $\langle 1,2,\ldots, m \rangle$, and $B_0$ is correctly computed as the sequence of the last characters $\langle s_1[k_1], \ldots, s_m[k_m] \rangle$ (i.e., $T_0$).

At the same time, the iteration $l$ computes the array $N_{l+1}$ to be used in the next iteration $l+1$ in order to compute the partial BWT $B_{l+1}$. Observe in fact that the i-$th$ $l$-suffix of $X_l$ is preceded by the symbol $c=T_l[q]$, where $q=N_l[i]$, and belongs to string $s_q$. Assuming that the i-$th$ suffix of $X_l$ is the h-$th$ suffix of $X_l$ which is preceded by that symbol $c$, then the $l+1$-suffix of $s_q$ is the h-$th$ suffix of $X_{l+1}$ starting with $c$. Furthermore, let us assume that there are $r$ $l$-suffixes of $X_l$ starting with a symbol smaller $c$. Then the $l+1$-suffix of $s_q$ is the (r+h)-$th$ suffix of $X_{l+1}$. By definition, it holds $N_{l+1}[r+h]=q$. The algorithm uses $\sigma+1$ lists $\mathcal{P}(\cdot)$, a list for each symbol in $\Sigma$, which are created at the beginning of iteration $l$. During the scanning of $N_l$ the index $q$ is added to the list $\mathcal{P}(c)$.

It is easy to prove that, at the end of iteration $l$, the concatenation of lists $\mathcal{P}(\cdot)$ (according to the order of the symbols in $\Sigma$) correctly gives $N_{l+1}$. Note that $N_{l+1}$ is computed also by the last iteration $k$, even though it is actually not used. Figure~\ref{fig:example-iteration-partialbwts} exemplifies the iteration $l=1$ of Algorithm~\ref{alg:sort-suffixes} which computes, for the set of reads of Figure~\ref{fig:example-vectors-B}, the partial BWT $B_1$ (see cycle for at line~\ref{alg:sort-suffixes:compute-projections-start}) and the array $N_2$ (line~\ref{alg:sort-suffixes:compute-N}), from the array $N_1$ (computed by the previous iteration $l=0$). Array $N_2$ will be used by the next iteration $l=2$ for computing the partial BWT $B_2$.

Observe that arrays $T_i$ must be kept in main memory, since they are not accessed sequentially (see Algorithm \ref{alg:ms-compute-partial}), and for this reason they cannot be stored in external memory.

\begin{algorithm2e}[htb!]
	\SetKwInOut{Input}{Input}\SetKwInOut{Output}{Output}
	\Input{The arrays $T_0, \ldots, T_{k}$}
	\For{$i\gets 1$ to $m$}{\label{alg:sort-suffixes:compute-N0}%
		$N_0[i] \gets i$\;
	}
	\For{$l\gets 0$ to $k$}{\label{alg:sort-suffixes:compute-partial}%
		\For{$t\gets 0$ to $\sigma$}{%
			$\mathcal{P}(c_t) \gets $ empty list\;
		}
		$B_l \gets $ empty list\;
		\For{$i\gets 1$ to $|N_{l}|$}{\label{alg:sort-suffixes:compute-projections-start}%
			$q \gets N_{l}[i]$\;
			$c \gets T_{l}[q]$\;
			\If{$c \ne \#$}{
				Append $c$ to $B_l$\;\label{alg:sort-suffixes:compute-B}
				Append $q$ to $\mathcal{P}(c)$\;
			}
		}
		$N_{l+1} \gets {\cal P}(c_0){\cal P}(c_1)\cdots{\cal P}(c_{\sigma})$\;\label{alg:sort-suffixes:compute-N}
	}
	\caption{Compute the partial BWTs $B_0, B_1, \ldots, B_k$}
	\label{alg:sort-suffixes}
\end{algorithm2e}

\begin{figure}[htbp]
	\renewcommand{\baselinestretch}{1.15}\normalsize
	(a) Input arrays $T_l$ ($0 \leq l \leq 4$) and array $N_1$\\[.5em]
	\begin{minipage}[c]{\textwidth}%
		\begin{minipage}{0.4\linewidth}%
			\begin{tabular}{|ccccc|}
				\hline
				$T_4$&$T_3$&$T_2$&$T_1$&$T_0$\\
				\hline
				\$&T&C&G&T\\
				\#&\#&\$&C&T\\
				\#&\$&A&C&A\\
				\hline
			\end{tabular}
		\end{minipage}%
		\begin{minipage}{0.55\linewidth}%
			\begin{tabular}{l}
				$N_{1} = \langle 3,1,2 \rangle$\\
			\end{tabular}
		\end{minipage}%
	\end{minipage}\\[1em]
	(b) Computing $B_1$ and $N_2$ (lines~\ref{alg:sort-suffixes:compute-projections-start}-\ref{alg:sort-suffixes:compute-N})\\[.5em]
	\begin{minipage}[c]{\textwidth}%
		\begin{minipage}{0.4\linewidth}%
			\begin{tabular}{|lc|}
				\hline
				$i=1$&\\\hline
				Read $3$ from $N_1[1]$&\\
				Read $C$ from $T_1[3]$&\\
				Append $C$ to $B_1$&\\
				Append $3$ to $\mathcal{P}($C$)$&\\
				\hline
			\end{tabular}
		\end{minipage}%
		\begin{minipage}{0.55\linewidth}%
			\begin{tabular}{ll}
				$\mathcal{P}(\$) = \langle  \rangle$&
				$\mathcal{P}($A$) = \langle  \rangle$\\
				$\mathcal{P}($C$) = \langle 3  \rangle$&
				$\mathcal{P}($G$) = \langle  \rangle$\\
				$\mathcal{P}($T$) = \langle  \rangle$&\\
				$B_1 = \langle C  \rangle$&\\
			\end{tabular}
		\end{minipage}\\[.5em]
		\begin{minipage}{0.4\linewidth}%
			\begin{tabular}{|lc|}
				\hline
				$i=2$&\\\hline
				Read $1$ from $N_1[2]$&\\
				Read $G$ from $T_1[1]$&\\
				Append $G$ to $B_1$&\\
				Append $1$ to $\mathcal{P}($G$)$&\\
				\hline
			\end{tabular}
		\end{minipage}%
		\begin{minipage}{0.55\linewidth}%
			\begin{tabular}{ll}
				$\mathcal{P}(\$) = \langle \rangle$&
				$\mathcal{P}($A$) = \langle \rangle$\\
				$\mathcal{P}($C$) = \langle 3 \rangle$&
				$\mathcal{P}($G$) = \langle 1 \rangle$\\
				$\mathcal{P}($T$) = \langle \rangle$&\\
				$B_1 = \langle C,G  \rangle$&\\
			\end{tabular}
		\end{minipage}\\[.5em]
		\begin{minipage}{0.4\linewidth}%
			\begin{tabular}{|lc|}
				\hline
				$i=3$&\\\hline
				Read $2$ from $N_1[3]$&\\
				Read $C$ from $T_1[2]$&\\
				Append $C$ to $B_1$&\\
				Append $2$ to $\mathcal{P}($C$)$&\\
				\hline
			\end{tabular}
		\end{minipage}%
		\begin{minipage}{0.55\linewidth}%
			\begin{tabular}{ll}
				$\mathcal{P}(\$) = \langle \rangle$&
				$\mathcal{P}($A$) = \langle \rangle$\\
				$\mathcal{P}($C$) = \langle 3,2 \rangle$&
				$\mathcal{P}($G$) = \langle 1 \rangle$\\
				$\mathcal{P}($T$) = \langle \rangle$&\\
				$B_1 = \langle C,G,C  \rangle$&\\
			\end{tabular}
		\end{minipage}\\[1em]
		\begin{minipage}{0.4\linewidth}%
			\begin{tabular}{|c|}
				\hline
				$N_2 \gets \mathcal{P}($C$)\mathcal{P}($G$)$ (line~\ref{alg:sort-suffixes:compute-N})\\\hline
			\end{tabular}
		\end{minipage}%
		\begin{minipage}{0.55\linewidth}%
			\begin{tabular}{l}
				$N_2 = \langle 3,2,1 \rangle$\\
			\end{tabular}
		\end{minipage}%
	\end{minipage}\\[1em]

	\caption{Example of iteration $l=1$ (computing $B_1$ and $N_2$ from $N_1$) of Algorithm~\ref{alg:sort-suffixes} for the set of strings presented in Figure~\ref{fig:example-vectors-B}. In (a) the three strings (excluding the \$) are depicted right-aligned inside a matrix. The sentinel \$ is placed immediately to the left of each string and the symbol \# pads the left empty space. Each $T_l$ is a column of the matrix. Array $N_1 = \langle 3,1,2\rangle$ gives the sorting of the $1$-suffixes of the input reads, that is,  suffix $A\$$ of the read number $3$, followed by suffix $T\$$ of the read number $1$ and finally suffix $T\$$ of the read number $2$. Array $N_2 = \langle 3,2,1\rangle$ gives the sorting of the $2$-suffixes: $CA\$$ of the read number 3, $CT\$$ of the read number 2, and $GT\$$ of the read number 1.
		Observe that lists $\mathcal{P}(\$), \mathcal{P}($A$), \mathcal{P}($T$)$ are empty during (and at the end of) the iterations of lines~\ref{alg:sort-suffixes:compute-projections-start}. Angle brackets are used for denoting both lists $\mathcal{P}(\cdot)$ and arrays $B_{1}$, $N_{1}$ and $N_2$. Indeed the latter three can be treated as lists since they are accessed sequentially.}
	\label{fig:example-iteration-partialbwts}
\end{figure}

\section{Backward strategy for computing the encoding $I_B$ and the LCP array}
\label{sec:merge-suffixes}
This section is devoted to describe the second step of our algorithm which computes the
BWT $B$ and the LCP array according to the backward strategy described in Section~\ref{sec:sketch-algorithm}.

First of all, we describe in detail how the single iteration works (see Algorithm~\ref{alg:compute-by-offset}). Then, we show how to enrich Algorithm~\ref{alg:compute-by-offset} in order to compute also the LCP array together with the encoding $I_X$ (see Algorithm~\ref{alg:compute-by-offset-with-lcp}). Finally, the complete procedure for computing $I_X=I_B$, from the partial BWTs $B_l$, is presented (see Algorithm~\ref{alg:merge-suffixes}) and is explained how to use the LCP array values in order to limit the iterations to the number strictly necessary to obtain $I_X$.

At this point, let us assume to have (iteration $p$) the encoding $I_{X^{p-1}}$ of the $p-1$-interleave $X^{p-1}$. We want to compute the encoding $I_{X^p}$ of the $p$-interleave $X^p$, by sorting the suffixes of $X^{p-1}$ by the first $p$ characters.
The algorithm implicitly obtains $X^{p}$ (suffixes sorted by the first $p$ characters) by implicitly reordering the characters preceding each one of the suffixes of $X^{p-1}$ (suffixes sorted by the first $p-1$ characters). We note that (by definition) for any $p$ from $0$ to $k$ the first $m$ entries of $I_{X^p}$ are all equal to $0$. Indeed, the $m$ $0$-suffixes (of the set $S$) occupy always the first $m$ positions for any value of $p$.

Before entering the details of iteration $p$ (see Algorithm~\ref{alg:compute-by-offset}), we give the idea of the algorithm. Let us consider the suffix $X^{p-1}[q]$ whose length is $l=I_{X^{p-1}}[q]$. Let $c$ be the symbol preceding such suffix.
Let $\mathcal{X}^{p-1}_s$ be the subset of suffixes of $X^{p-1}$ preceded by a symbol smaller than $c$, and let $\mathcal{X}^{p-1}_e$ be the subset of suffixes at a position $q' < q$ of $X^{p-1}$ preceded by the symbol $c$.
It is easy to show that the suffix $cX^{p-1}[q]$ (with length $l+1$) is greater (by the first $p$ characters) than all and only those suffixes $c_px$, such that $x \in \mathcal{X}^{p-1}_s \cup \mathcal{X}^{p-1}_e$ and $c_p$ is the symbol preceding $x$.
Therefore, the suffix $cX^{p-1}[q]$ (with length $l+1$) will occupy in $X^{p}$ the position $q'=|\mathcal{X}^{p-1}_s \cup \mathcal{X}^{p-1}_e|+1$, and $I_{X^p}[q']$ will be equal to $l+1$. In other words, position $q'$ of suffix $cX^{p-1}[q]$ on $X^p$ is given by the sum $h_s+h_e+1$ where $h_s$ is the number of suffixes of $X_{p-1}$ preceded by a symbol smaller than $c$ and $h_e$ is the number of suffixes, which are before $X^{p-1}[q]$ and preceded by symbol $c$.

Algorithm~\ref{alg:compute-by-offset} creates a set of $\sigma+1$ lists $\mathcal{I}(c_0), \mathcal{I}(c_1), \ldots, \mathcal{I}(c_{\sigma})$ containing at the end of iteration $p$ the partitioning of the encoding of $I_{X^{p}}$ by the first character $c_i$ of the suffixes of $X^p$.
Since the list $\mathcal{I}(c_0)$ (we recall that $c_{0}=\$$) is related to the  $0$-suffixes, then it is fixed over the iterations $p$ and is always composed of $m$ $0$s (and it is initialized at the beginning of the procedure). Therefore, at the end, the algorithm produces $I_{X^{p}}$ (see line~\ref{alg:compute-by-offset-concat}) as the concatenation $\mathcal{I}(c_0) \mathcal{I}(c_1) \cdots \mathcal{I}(c_{\sigma})$.
In order to fill the lists $\mathcal{I}(\cdot)$, Algorithm~\ref{alg:compute-by-offset} performs a scan of $I_{X^{p-1}}$. For each position $q$, it obtains $l=I_{X^{p-1}}[q]$, that is the length of the $q$-$th$ suffix of $X^{p-1}$, and the symbol $c$ preceding such suffix (see line~\ref{alg:compute-by-offset:get-symbol}). Vector $pos$ allows to read $c$ from the correct position of array $B_l$.
If $c \ne \$$, then $l$ is not greater than the length of the input string originating the suffix $X^{p-1}[q]$, %
and the integer $l+1$ is appended to the list $\mathcal{I}(c)$.
Otherwise, if $c = c_0 = \$$, it moves to the next position $q+1$. Indeed in this case, the value $l$ is greater than the length of the input string originating the suffix $X^{p-1}[q]$, thus the $cX^{p-1}[q]$ obtained is a $0$-suffix whose related integer $0$ should be appended to the list $\mathcal{I}(c_0)$, which is fixed (by definition) over the iterations.

This approach is alternative to the one presented in~\cite{Bonizzoni2016ANL} first and then implemented in \cite{jcb-bonizzoni-2019}. In fact, the iteration $p$ is a backward extension of the suffixes sorted by the first $p-1$ characters in order to obtain the suffixes sorted by the first $p$ characters. Instead the strategy presented in~\cite{Bonizzoni2016ANL} is based on a forward extension of the $(p-1)$-prefixes of the suffixes in order to obtain the ordering given by the encoding $I_{X^{p}}$.

The following theorem proves the correctness of Algorithm~\ref{alg:compute-by-offset}.

\begin{figure}[htbp]
	\renewcommand{\baselinestretch}{1.15}\normalsize
	(a) Input interleave $I_{X^0}$\\[.5em]
	\begin{minipage}[c]{\textwidth}%
		\begin{minipage}{.95\linewidth}%
			\begin{tabular}{l}
				$I_{X^0} = \langle 0,0,0,1,1,1,2,2,2,3,3,4 \rangle$\\
			\end{tabular}
		\end{minipage}%
	\end{minipage}\\[1em]
	\text{(b) Initialization of lists $\mathcal{I}(\cdot)$}\\[0.5em]
	\begin{minipage}[c]{\textwidth}%
		\begin{minipage}{.95\linewidth}%
			\begin{tabular}{ll}
				$\mathcal{I}(\$) = \langle 0,0,0 \rangle$&
				$\mathcal{I}($A$) = \langle  \rangle$\\
				$\mathcal{I}($C$) = \langle  \rangle$&
				$\mathcal{I}($G$) = \langle  \rangle$\\
				$\mathcal{I}($T$) = \langle  \rangle$&\\
			\end{tabular}
		\end{minipage}%
	\end{minipage}
	\vspace{0.5em}\\
	\text{(c) Scan of the interleave $I_{X^0}$ (lines~\ref{alg:compute-by-offset:for}-\ref{alg:compute-by-offset:for-end2})}\\[.5em]
	\begin{minipage}[t]{\textwidth}%
		\begin{minipage}[t]{.32\linewidth}%
			\begin{tabular}{|lc|}
				\hline
				$i=1$&\\\hline
				Read $0$ from $I_{X^0}$&\\
				Read $T$ from $B_0$&\\
				Append $1$ to $\mathcal{I}($T$)$&\\
				\hline
			\end{tabular}
		\end{minipage}%
		\begin{minipage}[t]{.32\linewidth}%
			\begin{tabular}{|lc|}
				\hline
				$i=2$&\\\hline
				Read $0$ from $I_{X^0}$&\\
				Read $T$ from $B_0$&\\
				Append $1$ to $\mathcal{I}($T$)$&\\
				\hline
			\end{tabular}
		\end{minipage}%
		\begin{minipage}[t]{.32\linewidth}%
			\begin{tabular}{|lc|}
				\hline
				$i=3$&\\\hline
				Read $0$ from $I_{X^0}$&\\
				Read $A$ from $B_0$&\\
				Append $1$ to $\mathcal{I}($A$)$&\\
				\hline
			\end{tabular}
		\end{minipage}\\[.5em]
		\begin{minipage}[t]{.32\linewidth}%
			\begin{tabular}{|lc|}
				\hline
				$i=4$&\\\hline
				Read $1$ from $I_{X^0}$&\\
				Read $C$ from $B_1$&\\
				Append $2$ to $\mathcal{I}($C$)$&\\
				\hline
			\end{tabular}
		\end{minipage}%
		\begin{minipage}[t]{.32\linewidth}%
			\begin{tabular}{|lc|}
				\hline
				$i=5$&\\\hline
				Read $1$ from $I_{X^0}$&\\
				Read $G$ from $B_1$&\\
				Append $2$ to $\mathcal{I}($G$)$&\\
				\hline
			\end{tabular}
		\end{minipage}%
		\begin{minipage}[t]{.32\linewidth}%
			\begin{tabular}{|lc|}
				\hline
				$i=6$&\\\hline
				Read $1$ from $I_{X^0}$&\\
				Read $C$ from $B_1$&\\
				Append $2$ to $\mathcal{I}($C$)$&\\
				\hline
			\end{tabular}
		\end{minipage}\\[.5em]
		\begin{minipage}[t]{.32\linewidth}%
			\begin{tabular}{|lc|}
				\hline
				$i=7$&\\\hline
				Read $2$ from $I_{X^0}$&\\
				Read $A$ from $B_2$&\\
				Append $3$ to $\mathcal{I}($A$)$&\\
				\hline
			\end{tabular}
		\end{minipage}%
		\begin{minipage}[t]{.32\linewidth}%
			\begin{tabular}{|lc|}
				\hline
				$i=8$&\\\hline
				Read $2$ from $I_{X^0}$&\\
				Read $\$$ from $B_2$&\\
				&\\
				\hline
			\end{tabular}
		\end{minipage}%
		\begin{minipage}[t]{.32\linewidth}%
			\begin{tabular}{|lc|}
				\hline
				$i=9$&\\\hline
				Read $2$ from $I_{X^0}$&\\
				Read $C$ from $B_2$&\\
				Append $3$ to $\mathcal{I}($C$)$&\\
				\hline
			\end{tabular}
		\end{minipage}\\[.5em]
		\begin{minipage}[t]{.32\linewidth}%
			\begin{tabular}{|lc|}
				\hline
				$i=10$&\\\hline
				Read $3$ from $I_{X^0}$&\\
				Read $\$$ from $B_3$&\\
				&\\
				\hline
			\end{tabular}
		\end{minipage}%
		\begin{minipage}[t]{.32\linewidth}%
			\begin{tabular}{|lc|}
				\hline
				$i=11$&\\\hline
				Read $3$ from $I_{X^0}$&\\
				Read $T$ from $B_3$&\\
				Append $4$ to $\mathcal{I}($T$)$&\\
				\hline
			\end{tabular}
		\end{minipage}%
		\begin{minipage}[t]{.32\linewidth}%
			\begin{tabular}{|lc|}
				\hline
				$i=12$&\\\hline
				Read $4$ from $I_{X^0}$&\\
				Read $\$$ from $B_4$&\\
				&\\
				\hline
			\end{tabular}
		\end{minipage}\\[.5em]
		\begin{minipage}[t]{.64\linewidth}%
			\begin{tabular}{ll}
				$\mathcal{I}(\$) = \langle 0,0,0 \rangle$&
				$\mathcal{I}($A$) = \langle 1,3  \rangle$\\
				$\mathcal{I}($C$) = \langle 2,2,3 \rangle$&
				$\mathcal{I}($G$) = \langle 2 \rangle$\\
				$\mathcal{I}($T$) = \langle 1,1,4  \rangle$&\\
			\end{tabular}
		\end{minipage}%
	\end{minipage}\\[1em]
	(d) Computing the interleave $I_{X^1}$ (line~\ref{alg:compute-by-offset-concat})\\[.5em]
	\begin{minipage}[c]{\textwidth}%
		\begin{minipage}[t]{.48\linewidth}%
			\begin{tabular}{|l|}
				\hline
				$I_{X^1} \gets \mathcal{I}(\$)\mathcal{I}($A$)\mathcal{I}($C$)\mathcal{I}($G$)\mathcal{I}($T$)$\\
				\hline
			\end{tabular}
		\end{minipage}%
		\begin{minipage}[t]{.48\linewidth}%
			\begin{tabular}{l}
				$I_{X^1} = \langle 0,0,0,1,3,2,2,3,2,1,1,4 \rangle$\\
			\end{tabular}
		\end{minipage}%
	\end{minipage}%
	\caption{Example of computing $I_{X^1}$ from $I_{X^0}$ (see Algorithm~\ref{alg:compute-by-offset}) for the set of reads presented in Figure~\ref{fig:example-vectors-B}. Angle brackets are used for denoting both lists $\mathcal{I}(\cdot)$ and arrays $I_{X^0}$, $I_{X^1}$. Indeed the latter two can be treated as lists since they are accessed sequentially. The two encodings $I_{X^1}$ and $I_{X^0}$ are those ones reported in Figure~\ref{fig:example-merge-suffixes}.}
	\label{fig:example-compute-interleave}
\end{figure}

\begin{theorem}
	\label{theorem:correctness1}
	If Algorithm~\ref{alg:compute-by-offset} receives in input the encoding $I_{X^{p-1}}$ of the $p-1$-interleave $X^{p-1}$, then it computes the encoding $I_{X^{p}}$ of the $p$-interleave $X^{p}$.
\end{theorem}

\begin{proof}

	Observe that the $(p-1)$-prefix (prefix with length $p-1$) of the i-$th$ suffix of
	$X_l$ is the suffix of the $p$-prefix of a suffix %
	of $X_{l+1}$,
	starting with the symbol $c=B_l[i]$.
	Then, line 7 of Algorithm~\ref{alg:compute-by-offset} appends length $l+1$ to the list $\mathcal{I}(c)$. Observe that line~\ref{alg:compute-by-offset:for-end2} implicitly computes a partitioning of the suffixes in $X^{p}$, according to their starting symbol, into lists $\mathcal{I}(c_0), \mathcal{I}(c_1), \ldots, \mathcal{I}(c_{\sigma})$, where $\mathcal{I}(c_i)$ gives the ordering, by the first $p$ characters, of the suffixes starting with symbol $c_i$. Each list $\mathcal{I}(c_i)$ contains (at line~\ref{alg:compute-by-offset-concat}) the lengths of such suffixes.

	Furthermore, given two distinct suffixes $c_{1}x_1$ and
	$c_{2}x_2$ such that  $c_{1}x_1$ is smaller (by the first $p-1$ characters) than $c_{2}x_2$,  either they begin with two different  symbols $c_{1} <  c_{2}$, or they both start with the same symbol, i.e., $c_{1} = c_{2}$.
	Let $L(c_{1})$ and $L(c_{2})$ be the partitions of $X^{p-1}$ containing the suffixes starting with $c_{1}$ and
	$c_{2}$ (respectively). Then, in $X^{p}$ all suffixes in $L(c_{1})$ precede those in $L(c_{2})$.
	Inside the list $L(c_{1})$, the ordering of two suffixes $c_{1}x_{i}$ and $c_{1}x_{j}$ by the first $p$ characters is the same as in $X^{p-1}$.
	Indeed, $cx_i[:p-1]$ is  lexicographically smaller than $cx_j[:p-1]$ if and only if $x_i[:p-1]$ is lexicographically smaller than $x_j[:p-1]$.
	It follows that $X^{p}$ consists of the concatenation of $L(c_i)$  according to the lexicographic ordering of symbols of  alphabet $\Sigma$, and thus line~\ref{alg:compute-by-offset-concat} of Algorithm~\ref{alg:compute-by-offset} computes the encoding $I_{X^{p}}$ of  $X^{p}$.
\end{proof}

\begin{algorithm2e}[tb!]
	\SetKwInOut{Input}{Input}\SetKwInOut{Output}{Output}

	$\mathcal{I}(c_0) \gets 0,0, \ldots, 0$\;
	\tcp{array of $m$ $0$s}
	$\mathcal{I}(c_1), \ldots, \mathcal{I}(c_{\sigma}) \gets $ empty lists\;

	$pos \gets$ integer array with $k+1$ $0$s\;%

	\For{$q\gets 1$ to $|I_{X^{p-1}}|$}{\label{alg:compute-by-offset:for}%
		$l \gets I_{X^{p-1}}[q]$\;\label{alg:compute-by-offset:for-start}
		$c \gets$ $B_l[pos[l]]$\;\label{alg:compute-by-offset:get-symbol}%
		$pos[l] \gets pos[l] + 1$\;%
		\If{$c \ne \$$}{
			Append $l+1$ to $\mathcal{I}(c)$\;\label{alg:compute-by-offset:for-end2}
		}
	}

	$I_{X^{p}} \gets \mathcal{I}(c_0) \mathcal{I}(c_1) \cdots \mathcal{I}(c_{\sigma})$\;\label{alg:compute-by-offset-concat}

	\caption{Compute $I_{X^{p}}$ from $I_{X^{p-1}}$}\label{alg:compute-by-offset}
\end{algorithm2e}

In the following we will describe how to compute the LCP array of the input dataset. Similarly to the computation of the BWT $B$, the LCP array will be constructed iteratively. More precisely, the LCP array will be constructed by considering prefixes of the suffixes by increasing length.
At this point, we can describe how to update Algorithm~\ref{alg:compute-by-offset} (iteration $p$) in order to compute (at the end of the iterations) also the LCP array.

To this aim we must introduce the following definition.

\begin{definition}
	\label{definition:p-LCP-array}
	Given the LCP array, $\mathit{LCP}_p$ is defined such that $\mathit{LCP}_p[i]=\min\{LCP[i],p\}$.
\end{definition}

Observe that $\mathit{LCP}_p[i]$ is the length of the longest prefix shared by the $p$-prefix of $X^p[i]$ and the $p$-prefix of $X^p[i-1]$.
We note that, when a suffix in $X^p$ is shorter than $p$, then its $p$-prefix (considered for $LCP_p$) is the whole suffix itself (\$ excluded).

The array $\mathit{LCP}_k$  is equal to the LCP array of the input set $S$, and
$\mathit{LCP}_0$ contains all $0$s, except for $\mathit{LCP}_0[1]$ that is equal to $-1$.
In Figure~\ref{fig:example-merge-suffixes} $\mathit{LCP}_0$ and $\mathit{LCP}_1$ are reported for the input set of Figure~\ref{fig:example-vectors-B}.

The LCP array is computed iteratively by starting from $\mathit{LCP}_0$. Now we describe the single iteration $p$ for computing $\mathit{LCP}_{p}$ from $\mathit{LCP}_{p-1}$.
Algorithm~\ref{alg:compute-by-offset-with-lcp} extends Algorithm~\ref{alg:compute-by-offset} in order to compute $I_{X^{p}}$ and $\mathit{LCP}_{p}$ from $I_{X^{p-1}}$ and $\mathit{LCP}_{p-1}$.

Algorithm~\ref{alg:compute-by-offset-with-lcp} builds  a set of $\sigma+1$ lists $\mathcal{L}(c_0), \mathcal{L}(c_1), \ldots, \mathcal{L}(c_{\sigma})$ containing the partitioning of the elements of $\mathit{LCP}_{p}$ by the first character $c_i$ ($0 \leq i \leq \sigma$) of the related suffix.
Since the list $\mathcal{L}(c_0=\$)$ is related to the $0$-suffixes, it is fixed for any iteration and is composed of $-1$ followed by $m-1$ $0$s. Moreover, observe that the first element of each list $\mathcal{L}(c_i)$ ($1 \leq i \leq \sigma$) is always $0$.
Finally, Algorithm~\ref{alg:compute-by-offset-with-lcp} concatenates all the lists $\mathcal{L}(c_0), \mathcal{L}(c_1), \ldots, \mathcal{L}(c_{\sigma})$, thus producing $\mathit{LCP}_{p}$ (see line~\ref{alg:compute-by-offset-with-lcp-concat}).

Before giving the detail of computing the single lists $\mathcal{L}(\cdot)$, we need to introduce the following function.
Given a position $q$ and a symbol $c \ne \$$, the function $\alpha_p(q,c)$ is the length of the longest prefix shared by the $p$-prefixes of suffixes $X^p[q]$ and $X^p[h]$ where $h$ is the biggest position before $q$ related to a suffix $X^p[h]$ preceded by symbol $c$. If  such $h$ does not exist, then $\alpha_p(q,c) = -1$.

In the following, given two strings $x_1, x_2$, we denote (respectively) by $lcp_p(x_1, x_2)$ and $lcp(x_1, x_2)$ the length of the longest common prefix between the $p$-prefixes of $x_1$ and $x_2$, and the length of the longest common prefix between $x_1$ and $x_2$ (that is, $lcp(x_1, x_2)=lcp_k(x_1,x_2)$).
The following proposition relates the values of $\alpha_{p-1}(q,c)$ and $\mathit{LCP}_{p}$ and it  is a direct  consequence of their definitions.

\begin{proposition}
	\label{proposition:lcp-alpha}
	Let $cx_{1}$  and $cx_{2}$ be two consecutive suffixes of $X^{p}$, and let $x_{2}$ be the q-$th$ suffix of $X^{p-1}$.
	Then $\min\{p, lcp(cx_{1}, cx_{2})\} = 1 + \alpha_{p-1}(q,c)$.
\end{proposition}

During the scan of the encoding $I_{X^{p-1}}$, the value $\mathit{LCP}_{p-1}[i]$ is obtained (see line 13 of Algorithm \ref{alg:compute-by-offset-with-lcp}).
The function $\alpha_{p-1}(q,c)$ is maintained in the array $\alpha$ of size $\sigma-1$ initially set to $\sigma-1$ values $-1$s%
, and updated in the cycle at line~\ref{alg:compute-by-offset-with-lcp-update-alpha}. The main invariant of Algorithm~\ref{alg:compute-by-offset-with-lcp} is that, at
line~\ref{alg:compute-by-offset-with-lcp:compute-lcp}, the variable $\alpha[c]$ is equal
to  $\alpha_{p-1}(q, c)$---this invariant is a consequence of the following
Lemma~\ref{lemma:computing-alpha} and can be proved by a direct inspection of Algorithm~\ref{alg:compute-by-offset-with-lcp}. The value $\alpha[c]$ incremented by $1$ is appended to the list $\mathcal{L}(c)$.

\begin{lemma}
	\label{lemma:computing-alpha}
	Let $x_{1}$  and $x_{2}$ be respectively the j-$th$ and the q-$th$ suffixes of $X^{p-1}$, such that $j<q$, and let $c$ be the  symbol  preceding  suffix $x_{1}$.
	If  every suffix at a position $t$ between $j$ and $q$ ($j<t<q$), is not preceded by the symbol $c$,  then it holds that $\alpha_{p-1}(q,c) = \min_{j< h\le q}\{\mathit{LCP}_{p-1}[h]\}$.
\end{lemma}

\begin{proof}
	Since $c$ is not the symbol that precedes the suffix at position $t$ \ with $j<t<q$, then by definition of $\alpha_{p-1}(q,c)$, it must be that $\alpha_{p-1}(q,c)= lcp_{p-1}(X^{p-1}[j], X^{p}[q])$, since the $j$  is the largest integer less than $q$ for which the j-$th$ suffix is preceded by symbol $c$.
	Since it is immediate to verify that $lcp_{p-1}(X^{p}[j], X^{p}[q])= \min_{j< h\le q}\{\mathit{LCP}_{p-1}[h]\}$, the lemma easily follows.
\end{proof}

The previous argument allows us to prove the following theorem which, combined with
Theorem~\ref{theorem:correctness1} completes the correctness of Algorithm~\ref{alg:compute-by-offset-with-lcp}.

\begin{theorem}
	\label{theorem:correctness2}
	Given as input $\mathit{LCP}_{p-1}$ and the partial BWTs $B_{0}, B_1, \ldots, B_k$,
	Algorithm~\ref{alg:compute-by-offset-with-lcp} computes $\mathit{LCP}_{p}$.
\end{theorem}

\begin{proof}
	Observe that $\alpha[c]\ge 0$ at line~\ref{alg:compute-by-offset-with-lcp:compute-lcp} iff
	the current suffix at position $q$ is not the first to be preceded by the character $c$, hence we must append the value $1 + \alpha_{p-1}(q,c)$ to $\mathcal{L}(c)$.
	Since $\alpha[c] = \alpha_{p-1}(q,c)$, the theorem is proved.
\end{proof}

\begin{algorithm2e}[tb!]
	\SetKwInOut{Input}{Input}\SetKwInOut{Output}{Output}

	$\mathcal{I}(c_0) \gets 0,0, \ldots, 0$\;
	\tcp{vector of $m$ $0$s}
	$\mathcal{I}(c_1), \ldots, \mathcal{I}(c_{\sigma}) \gets $ empty lists\;
	\tcp{array of $m-1$ $0$s preceded by $-1$}
	$\mathcal{L}(c_0) \gets -1,0, \ldots, 0$\;
	$\mathcal{L}(c_1), \ldots, \mathcal{L}(c_{\sigma}) \gets $ empty lists\;

	$pos \gets$ integer array with $k+1$ $0$s\;%

	\ForEach{$c\in \{c_1, \ldots, c_{\sigma}\}$}{%
		Append $0$ to $\mathcal{L}(c)$\label{alg:compute-by-offset-with-lcp:init2}\;
		$\alpha[c] \gets -1$\;
	}

	\For{$q\gets 1$ to $|I_{X^{p-1}}|$}{%
		$l \gets I_{X^{p-1}}[q]$\;
		$c \gets$ $B_l[pos[l]]$\;%
		$pos[l] \gets pos[l] + 1$\;%
		\If{$c \ne \$$}{
			Append $(l+1)$ to $\mathcal{I}(c)$\;\label{alg:compute-by-offset-with-lcp:append}
		}

		$lcp \gets \mathit{LCP}_{p-1}[i]$\;

		\ForEach{$d\in \{c_1 \ldots, c_{\sigma}\}$}{\label{alg:compute-by-offset-with-lcp-update-alpha}
			$\alpha[d] \gets \min\{\alpha[d], lcp\}$\;\label{alg:compute-by-offset-with-lcp-update-alpha-end}
		}

		\If{$c \ne \$$ and $\alpha[c] \ge 0$\label{alg:compute-by-offset-with-lcp:compute-lcp}}{%
			Append $\alpha[c] + 1$ to $\mathcal{L}(c)$\;
			$\alpha[c] = \infty$\;
		}

	}
	$I_{X^{p}} \gets \mathcal{I}(c_0) \mathcal{I}(c_1) \cdots \mathcal{I}(c_{\sigma})$\;
	$\mathit{LCP}_{p} \gets \mathcal{L}(c_0) \mathcal{L}(c_1) \cdots \mathcal{L}(c_{\sigma})$\;\label{alg:compute-by-offset-with-lcp-concat}

	\caption{Compute $I_{X^{p}}$ and $\mathit{LCP}_{p}$ from $I_{X^{p-1}}$ and $\mathit{LCP}_{p-1}$}
	\label{alg:compute-by-offset-with-lcp}
\end{algorithm2e}

In Figure~\ref{fig:example-compute-interleave} the computation of $I_{X^1}$ from $I_{X^0}$ (by Algorithm~\ref{alg:compute-by-offset} for $p=1$) is shown for the set $S$ of reads presented in Figure~\ref{fig:example-vectors-B}. The encodings $I_{X^1}$ and $I_{X^0}$ are reported in Figure~\ref{fig:example-merge-suffixes} together with $LCP_1$ and $LCP_0$ whose computation (by Algorithm~\ref{alg:compute-by-offset-with-lcp}) has been omitted for simplicity.

The procedure \textbf{BWT+LCP} (see Algorithm~\ref{alg:merge-suffixes}) computes $I_{X^k}$ and $\mathit{LCP}_k$, which are the encoding of the BWT and the LCP array of the input set $S$ of strings, by iterating Algorithm~\ref{alg:compute-by-offset-with-lcp}. Iterations stop when the maximum value $\max_{q}\{\mathit{LCP}_{p}[q]\}$ in the array $\mathit{LCP}_{p}$ is less than $p$. In fact,  it means that for an iteration  $t > p$, the values $I_{X^{t}}$ and $\mathit{LCP}_{t}$ do not change since the suffixes have been fully sorted and thus $I_{X^t}$ and $\mathit{LCP}_t$ remain equal to $I_{X^k}$ and $\mathit{LCP}_k$, respectively.
The correctness of the procedure \textbf{BWT+LCP}
is a consequence of Theorem~\ref{theorem:correctness2} and Definition~\ref{definition:p-LCP-array}. Observe that if the maximum value in the LCP array is equal to $z$, then at each iteration $p$ of Algorithm~\ref{alg:merge-suffixes} with $p \leq z$, the maximum value in $\mathit{LCP}_{p}$ is $p$, in virtue of Theorem~\ref{theorem:correctness2} and Definition~\ref{definition:p-LCP-array}. When $p = z+1$, then by Definition~\ref{definition:p-LCP-array}, the iteration $p$ gives value $z$, that is $\max_{q}\{\mathit{LCP}_{p}[q]\} < p$.  Then the suffixes have been fully sorted and the LCP array has been computed at the previous step $p = z$.

\begin{algorithm2e}[htb!]
	\SetKwInOut{Input}{Input}\SetKwInOut{Output}{Output}
	\Input{The strings $s_{1}, \ldots , s_{m}$, and their maximum length $k$}
	\Output{The BWT $B$ and the LCP array of the input strings}

	Compute $T_{0}, \ldots , T_{k}$ from $s_{1}, \ldots , s_{m}$\;\label{algo:ms-prepare-t}

	Apply Algorithm~\ref{alg:sort-suffixes} to compute $B_{0}, \ldots B_{k}$\;\label{alg:ms-compute-partial}
	\For{$1\le i\le \sum_{j=0}^{k}|B_{j}|$}{%
		$h\gets$ the smallest integer such that $\sum_{j=0}^{h}|B_{j}| \ge i$\;
		$I_{X^0}[i] \gets h$\;
	}
	\tcc{Informally, $I_{X^0}$ is the array made of $|B_0|$ $0$s, $|B_1|$ $1$s, $\ldots$, $|B_k|$ $k$s}
	$\mathit{LCP}_0 \gets -1, 0, 0, \ldots, 0$\;

	$p \gets 1$\;

	\While{\(\mathtt{TRUE}\)}{\label{alg:merge-suffixes:iteration-condition}%
		Apply Algorithm~\ref{alg:compute-by-offset-with-lcp} to compute $I_{X^{p}}$
		and $\mathit{LCP}_{p}$ from $I_{X^{p-1}}$, $\mathit{LCP}_{p-1}$\; %
		\If{$\max_{q}\{\mathit{LCP}_{p}[q]\} \not = p$}{\textbf{break}\;}
		$p \gets p+1$\;\label{alg:ms-merge-suffixes-end}
	}

	Reconstruct $B$ from $I_{X^{p-1}}$ and $B_0, \ldots, B_k$\;
	\Output{($B$, $\mathit{LCP}_{p-1}$)}
	\caption{BWT+LCP}\label{alg:merge-suffixes}
\end{algorithm2e}

\begin{figure}[htbp]
	\centering\renewcommand{\baselinestretch}{1.15}\normalsize
	\begin{tabular}{|ccl|ccl|}
		\hline
		$I_{X^0}$&$\mathit{LCP}_0$&$X^0$&$I_{X^1}$&$\mathit{LCP}_1$&$X^1$\\\hline
		0&-1&\$&0&-1&\$\\
		0&0&\$&0&0&\$\\
		0&0&\$&0&0&\$\\
		1&0&A\$&1&0&A\$\\
		1&0&T\$&3&1&ACA\$\\
		1&0&T\$&2&0&CA\$\\
		2&0&CA\$&2&1&CT\$\\
		2&0&CT\$&3&1&CGT\$\\
		2&0&GT\$&2&0&GT\$\\
		3&0&ACA\$&1&0&T\$\\
		3&0&CGT\$&1&1&T\$\\
		4&0&TCGT\$&4&1&TCGT\$\\
		\hline
	\end{tabular}
	\caption{Encodings $I_{X^{0}}$ and $I_{X^{1}}$, and arrays $\mathit{LCP}_{0}$ and $\mathit{LCP}_{1}$ for the set of reads presented in Figure~\ref{fig:example-vectors-B}, together with their related interleaves $X^0$ and $X^1$. This figure depicts the encoding $I_{X^1}$ and the array $LCP_1$ (on the right) which are computed from $I_{X^0}$ and $LCP_0$ (on the left) by the iteration $p=1$ of Algorithm~\ref{alg:compute-by-offset-with-lcp}.}
	\label{fig:example-merge-suffixes}
\end{figure}

Observe that, in virtue of the radix sort strategy, the two steps of our method (computing the partial BWTs and computing the interleave $I_X$) do not depend on the particular order of the strings in the input set $S$. For this reason, there is no particular order of the input strings which may improve the computation.

\subsection{Comparison with other strategies}
While a common element of our method with \egap and \beetl is the use of a radix sort strategy, a main difference is represented by the collection of objects to which it is applied.
\beetl's algorithm works by a unique step and is based on the following invariant: at the iteration $p$, it computes the partial BWT for the collection of suffixes of length at most $p$. Differently, our algorithm works by two steps: first it computes the partial BWTs $B_l$ (as
previously defined) and then the interleave $I_X$. In the second step the following invariant is maintained: at the iteration $p$, it computes the list of the symbols preceding all the suffixes in the input collection $S$ sorted by the $p$-long prefixes.
As a consequence, at the iteration $p$, it computes a permutation of the BWT for $S$ tending to the solution over the iterations, while \beetl computes a subsequence of the BWT for $S$ and  maintains over the iterations the reciprocal order between the symbols.

Arrays $N_l$ used by the first step of our algorithm (computing the partial BWTs $B_l$) are the same as arrays $N_l$ used in \cite{bauer_lightweight_2013}. Indeed,
$N_l[i]$ in our case is the position in $S$ of the string which is the origin of the $i$-th $l$-suffix $X_l[i]$ whose preceding symbol is $B_l[i]$.
Arrays $N_l(h)$ in \cite{bauer_lightweight_2013} are defined such that $N_l(h)[i]$ is the position in $S$, of the string which is the origin of the $i$-th $l$-suffix (in the partial BWT) starting with the $h$-th symbol $c_h$ of the alphabet. We note that the concatenation $N_l(0)  N_l(1) ... N_l(\sigma)$ gives the array $N_l$ of our algorithm.

Observe that both our algorithm and \egap use the notion of an interleave in order to compute the BWT and the LCP array. More precisely,  \egap splits the input collection $S$ into subcollections sufficiently small, then it computes the BWT (partial BWTs) for each subcollection and finally it merges the BWTs similarly to the approach in~\cite{Holt2014}.
On the other hand, our algorithm first computes the partial BWTs $B_l$ from the whole collection $S$, that are then merged maintaining the invariant property described above.

\section{Complexity}
\label{sec:complexity}

In this section we will analyze the computational and I/O volume of our algorithm. %

First we will analyze Algorithm~\ref{alg:sort-suffixes}.
This procedure mainly consists of two
nested loops in which each operation requires constant time.
If the input is a set of \(m\) strings of length \(k\), the time complexity of it is \(\mathcal{O}(mk)\).
Note that each of the $k+1$ lists $B_l$ and $N_l$ have $m$ elements which are read or written sequentially and,
moreover, each list is read only once per execution.
Hence, the I/O volume of
Algorithm~\ref{alg:sort-suffixes} is \(\mathcal{O}(mk\lg m)\) since, for each element in
\(T_0, \ldots, T_k\), Algorithm~\ref{alg:sort-suffixes} appends an integer
less than \(m\) to the correct list \(\mathcal{P(\cdot)}\) that we can store
on disk, since we access them sequentially.

Besides some $\mathcal{O}(1)$-space data structures, the algorithm uses $\sigma + 1$ lists
$\mathcal{P}(\cdot)$ to store pointers to the open files and $k+1$ arrays $T_0, T_1, \ldots, T_k$ to store
the characters of the sequences.
Note that, at each iteration of the loop at
line~\ref{alg:sort-suffixes:compute-partial}, only one array $T_{l}$ must be kept in main memory, since we need to perform non-sequential
accesses, and requires \(m\lg\sigma\) bits---notice that for one million DNA reads, that
translates to 256 Mbytes of memory, which is well below the RAM amount found in standard PCs.
Therefore, if we can address each file using \(w\) bits, the main memory requirement of Algorithm~\ref{alg:sort-suffixes} is \(\mathcal{O}(\sigma w + k w + m\lg \sigma)\) bits.

Furthermore, arrays $T_l$ ($0 \leq l \leq k$) can be computed in $\mathcal{O}(km)$ time and $\mathcal{O}(km\lg \sigma)$ I/O volume by reading sequentially $\mathcal{B}$ input strings at a time and producing $\mathcal{B}$
positions of arrays $T_{l}$, where $\mathcal{B}$ is the disk block size, that is the
number of characters that are read or written in a single disk operation (see
\cite{Vitter:2001:EMA:384192.384193}).
Notice that this step requires to keep $l\times\mathcal{B}$ characters in main memory: this is
not a problem for Bioinformatics applications, since short reads are at most a
few hundreds of characters long, and even longer reads are at most 20000
characters long. Anyway, it is possible to adapt the algorithm of~\cite{DBLP:journals/cacm/AggarwalV88,vitter-external-memory} to compute the arrays $T_{l}$ arrays with $\mathcal{O}(\mathcal{B}^{2})$ main memory.

We will now analyze Algorithm~\ref{alg:compute-by-offset-with-lcp}.
The time complexity of this procedure is
\(\mathcal{O}(mk\sigma)\) since such procedure is composed of a \texttt{for}
loop that iterates over the encoding  \(I_{X^{p-1}}\)---whose length
is \(mk\)---performing constant time operations per element except
for the loop at
lines~\ref{alg:compute-by-offset-with-lcp-update-alpha}--\ref{alg:compute-by-offset-with-lcp-update-alpha-end}
that requires \(\mathcal{O}(\sigma)\) time.

The I/O volume is $\mathcal{O}(mk\max\{\lg m, \lg l\})$ bits, since each iteration of the loop at
lines~\ref{alg:compute-by-offset-with-lcp-update-alpha}--\ref{alg:compute-by-offset-with-lcp-update-alpha-end}
requires to read and write a constant number of elements of some lists whose values are
bounded by $m$ or $l$, and since \(\alpha\) %
is kept in main memory.
The main memory usage is $\mathcal{O}(\sigma\lg l + kw)$ bits, since we store \(\sigma\) integers smaller than \(l\) in \(\alpha\)
and \(k\) pointers to the lists $B_{i}$.

We can now analyze Algorithm~\ref{alg:merge-suffixes}, which is composed of
two main steps:
in the first one it prepares the input data structures
(line~\ref{algo:ms-prepare-t}), invokes Algorithm~\ref{alg:sort-suffixes}, and initializes
some data structures.
In the second part
(lines~\ref{alg:merge-suffixes:iteration-condition}--\ref{alg:ms-merge-suffixes-end})
it computes the final encoding \(I_{X^{p}}\) and the LCP array from the
structures computed at the previous step by iteratively applying
Algorithm~\ref{alg:compute-by-offset-with-lcp}.

The complexity of the first part is essentially that of
Algorithm~\ref{alg:sort-suffixes}, since computing the lists \(T_0, \ldots , T_k\)
(line~\ref{algo:ms-prepare-t}) requires \(\mathcal{O}(mk)\) with
a single scan of the input data (whose size is \(mk\)), while
outputting the lists requires constant time
per element.

The second step is mainly composed of a  \texttt{while} loop that iteratively
applies Algorithm~\ref{alg:compute-by-offset-with-lcp} (that requires \(\mathcal{O}(mk\sigma)\)) to compute the final
interleave and the final LCP array.
Moreover, the proof of correctness of Algorithm~\ref{alg:merge-suffixes} also shows that
Algorithm~\ref{alg:compute-by-offset-with-lcp} is applied $l+1$ times,  where
\(l\) is the largest value in the LCP array.

Finally, Algorithm~\ref{alg:merge-suffixes} builds the final BWT from \(I_{X^P}\)
and the lists \(B_0, \ldots, B_k\) by a single scan of those $m$-long lists, which requires
\(\mathcal{O}(mk)\) time overall.
Therefore, Algorithm~\ref{alg:merge-suffixes} requires an overall
\(\mathcal{O}(mkl\sigma)\) time.

The I/O complexity of the first step is \(\mathcal{O}(\max\{ mk\lg m, mk\lg \sigma\})\) bits whereas the main
memory requirement is \(\mathcal{O}(\sigma w + k w + m\lg \sigma)\) bits.
Indeed, computing the lists \(T_0, \ldots, T_k\) at line~\ref{algo:ms-prepare-t}
requires us to store only one character per time of each sequence \(s_i\) and to
append it to the correct list: therefore it has \(\mathcal{O}(mk\lg \sigma)\) bits I/O
volume and \(\mathcal{O}(kw + \lg \sigma)\) bits main memory requirement.
We have to include the requirements of Algorithm~\ref{alg:sort-suffixes}, which changes
the main memory needed for the first step to \(\mathcal{O}(\sigma w + k w + m\lg \sigma)\) bits.

The I/O volume of the second step is \(\mathcal{O}(mkl\lg l)\) bits since it consists
essentially of $l$ applications of Algorithm~\ref{alg:compute-by-offset-with-lcp}.
Finally, while building the final BWT from \(I_{X^P}\),
Algorithm~\ref{alg:merge-suffixes} reads \(\mathcal{O}(mk\lg m)\) bits due to
the interleave and \(\mathcal{O}(mk\lg \sigma)\) bits due to the partial BWTs,
writes \(\mathcal{O}(mk\lg\sigma)\) bits for the final BWT and requires
\(\mathcal{O}(\max\{\lg l, \lg \sigma\})\) bits of main memory since at most it stores in main memory one
element of \(I_{X^P}\) and one element of a partial BWT.

Therefore, overall Algorithm~\ref{alg:merge-suffixes} reads and writes
$\mathcal{O}(mkl\max\{\lg m, \lg l\})$ bits from and to the disk and requires
\(\mathcal{O}(\sigma w + k w + m\lg \sigma + \lg l)\) bits of main memory.
We can summarize our results as follows.

\begin{proposition}
	\label{proposition:complexity}
	Given as input a set composed of \(m\) strings of length \(k\) over and alphabet
	of size \(\sigma\), the procedure \textbf{BWT+LCP} computes the BWT and the LCP
	array of it in \(\mathcal{O}(mkl\sigma)\) time, where \(l\) is the maximal value of
	the LCP array.
	This procedure requires to store in main memory \(\mathcal{O}(\sigma w + k w + m \lg \sigma + \lg l)\) bits
	and reads and writes from and to the disk \(\mathcal{O}(mkl\max\{\lg m, \lg l\})\) bits.
\end{proposition}

Note that, if \(\sigma\) is constant then the time complexity of the method presented in
this paper becomes \(\mathcal{O}(mkl)\).
Moreover, if the word size is \(\max\{w, \lg m, \lg l \}\) then its I/O
volume and main memory requirement become \(\mathcal{O}(mkl)\) and
\(\mathcal{O}(k + m)\) respectively.

\section{Results}
\label{sec:exp-eval}

We implemented the method proposed in this article in a
prototype in C, named \texttt{bwt-lcp-em} (we will refer to it as \bwtlcpem in the following) that is freely available at
\url{https://github.com/AlgoLab/bwt-lcp-em}.  We compared our method
with other tools specifically designed to index datasets composed by a
huge number of short sequences such as Next-Generation Sequencing
read sets.

We have compared \bwtlcpem with the original
implementation of \extlcp~\cite{Cox2016} (\beetl), as well as a more recent version (\beetlem) that implements a fully external memory approach\footnote{The second version is available at \url{https://github.com/giovannarosone/BCR_LCP_GSA}.}, the in-memory
method \gsais~\cite{DBLP:journals/tcs/LouzaGT17}, and two recent external memory
tools (\egsa~\cite{Louza2017d} and \egap~\cite{DBLP:journals/almob/EgidiLMT19}).
Notice that \gsais and \egsa have been designed to compute the suffix array of
a set of strings, so the computation of the BWT and of the LCP array is likely
not optimized.

We used some non-default values for some of the parameters in order to
minimize main memory usage: \egap has been run with \verb+--lbytes 1+, \beetl
with \verb+--memory-limit=900+, and \egsa by compiling with \verb+MEMLIMIT=900+.
We allowed \egap to use 95\% of the available RAM, as suggested in its website.

We compared the considered tools in the scenario of 1GB of main memory
available, by considering instances with 1, 2, 4, 8, 16, and 32 million
sequences, taken from two different data sources: (a) 148bp Illumina reads from the Genome In A
Bottle (GIAB)~\cite{giab2014} consortium, more precisely from the NA24385
individual; (b) random sequences of length $151$ generated by a \texttt{Python} script that builds uniformly distributed fixed-length sequences over the DNA alphabet (for the extended and detailed experimentation results we refer the read to the GitHub repository of the implementation).
The goal of using these two datasets is to experimentally assess the theoretical
time complexity of our approach that shows a dependency on $l$, i.e.,
the maximum value of the LCP array (see Section~\ref{sec:complexity}).
We expect that $l$ in the random datasets will be considerably less than the
length $k$ of input strings.
On the other hand, the Illumina datasets represent the worst case of our
approach as they have a 300$\times$ mean coverage.
Hence they will surely contain duplicate reads and $l$ will be equal to the
length $k$ of the input strings.

We ran all the experiments on the same workstation running Ubuntu Linux
18.04 equipped with an Intel Core i7-4770 CPU running at 3.40GHz and a 256GB solid state disk. The machine is equipped with 8GB of RAM, and we limited the amount of RAM at boot time to 1GB to avoid the effects of OS caching.

Table~\ref{table:experimental-evaluation-times-1g} reports the time (in
minutes) required to compute the BWT and the LCP array. The first column
indicates the number of sequences in the dataset, whereas column $l$ indicates
the maximum value of the LCP array on that dataset.
Symbol \(\star\) means that
the tool could not complete the execution in our environment because it needed
more than 1GB of RAM, while symbol \(\diamond\) means that the tool could not
complete the execution since it required more disk space than that available.
Notice that \beetl and \gsais did not complete some executions because of the
RAM limit, while \egsa required more disk space than that available.
To better highlight the trends,
Figure~\ref{fig:experimental-evaluation-times-1g} visually depicts the same results
presented in the table.
We expect that this experiment will show the advantages of memory-conscious
approaches.

\begin{table}[t]
  \centering\renewcommand{\baselinestretch}{1.15}\small
  \begin{tabular}{r >{\raggedleft\arraybackslash}p{2em}rrrrrr}
    \toprule
    \multirow{2}{4em}{\centering No.~of strings}
    & \multicolumn{7}{c}{Dataset: NA24385} \\
    \cmidrule{2-8}
     & $l$ & \bwtlcpem & \beetlem & \beetl & \gsais & \egsa & \egap \\
    \midrule
    1M & 148 & 26   & 12  & 11       & 28       & 60        & 1\\
    2M & 148 & 53   & 25  & 23       & \(\star\)& 239       & 18    \\
    4M & 148 & 105  & 51  & 52       & \(\star\)& 693       & 44    \\
    8M & 148 & 213  & 122 & 124      & $\star$  & 2370      & 184 \\
   16M & 148 & 414  & 241 & 227      & $\star$  & \(\diamond\) & 448  \\
   32M & 148 & 855  & 633 & \(\star\)& $\star$  & \(\diamond\) & 915    \\
  \end{tabular}

  \begin{tabular}{r >{\raggedleft\arraybackslash}p{2em}rrrrrr}
    \toprule
    \multirow{2}{4em}{\centering No.~of strings}
    & \multicolumn{7}{c}{Dataset: random} \\
    \cmidrule{2-8}
     & $l$ & \bwtlcpem & \beetlem & \beetl & \gsais & \egsa & \egap\\
    \midrule
    1M & 25  & 5    & 13  & 12      & 29       & 10        & 1  \\
    2M & 28  & 11   & 27  & 25      & \(\star\)& 25        & 11 \\
    4M & 27  & 21   & 50  & 54      & \(\star\)& 49        & 24 \\
    8M & 28  & 43   & 104 & 119     & $\star$  & 138       & 92\\
   16M & 30  & 108  & 249 & 263     & $\star$  & \(\diamond\) & 227\\
   32M & 32  & 238  & 813 &\(\star\)&$\star$   & \(\diamond\) & 434\\
    \bottomrule
  \end{tabular} \caption{Time required to compute the BWT and the LCP array (in minutes) on the NA24385 and random datasets using a PC with 1GB RAM.
    The first column indicates the number of sequences in the dataset.
    Column $l$ indicates the maximum value of the LCP array on that dataset.
    \(\star\) means that the tool required more than 1GB of RAM.
    \(\diamond\) means that the tool crashed because of disk space exhaustion.
  }
  \label{table:experimental-evaluation-times-1g}
\end{table}

\begin{figure}[t]
  \centering
  \includegraphics[width=0.99\linewidth]{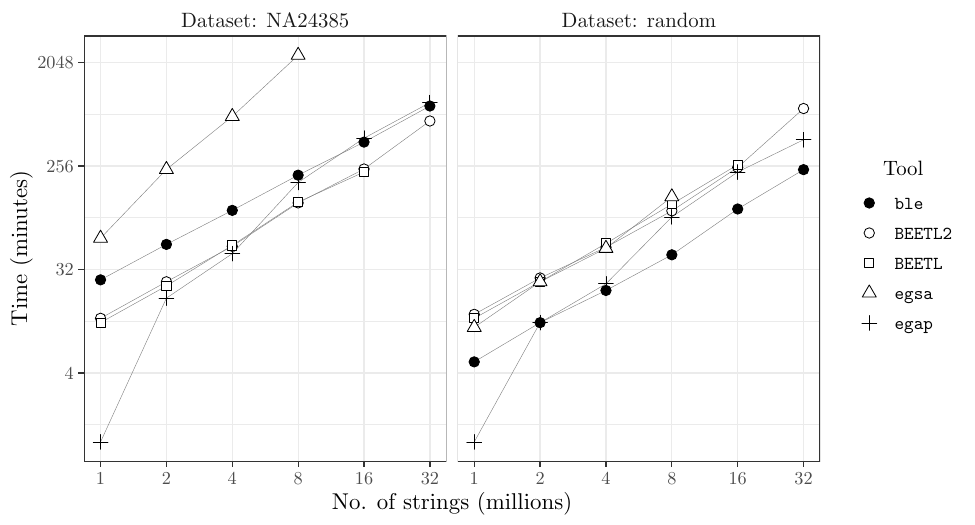}
  \caption{Time required to compute the BWT and the LCP array (in minutes) on the NA24385 and random datasets using a PC with 1GB RAM.}
  \label{fig:experimental-evaluation-times-1g}
\end{figure}

The results point out that only \bwtlcpem, \beetl, \beetlem, and \egap were able to deal with
such a limited amount of main memory. Moreover, \bwtlcpem, \beetlem, and \egap were able to
compute the BWT and LCP array for all the datasets, while \gsais and \egsa could
only cope with smaller datasets.
Since \gsais is an in-memory approach, its memory requirements made impossible to process even moderately large instances.

On the NA24385 dataset, \egap is the fastest tool on the instances up to 8M
reads, while \beetlem is the fastest on larger instances.
Still, the trend in the running time hints that \bwtlcpem will likely become the
fastest tool on instances larger than those considered here.

On the random dataset, \bwtlcpem is the fastest tool on all instances with at
least 2 million reads---\egap is the fastest on 1 million reads.
\egap and \bwtlcpem are always the two fastest tools.

The comparison of the running times on the two datasets empirically confirms
that \bwtlcpem has a time complexity that depends linearly on the maximum value
of the LCP array, while \beetl and \beetlem depend only on the maximum length of
the input strings.
As expected, also \egsa and \egap show a dependency on the maximum value of the LCP
array, as both are definitely faster on the random dataset than on the NA24385
dataset (for \egap roughly \(2\times\), for \egsa from \(6\times\) to \(15\times\)).

\begin{table}[t]
  \centering\renewcommand{\baselinestretch}{1.15}\small
  \begin{tabular}{r >{\raggedleft\arraybackslash}p{2em}rrrrrr}
    \toprule
    \multirow{2}{4em}{\centering No.~of strings}
    & \multicolumn{7}{c}{Dataset: NA24385} \\
    \cmidrule{2-8}
     & $l$ & \bwtlcpem & \beetlem & \beetl & \gsais & \egsa & \egap \\
    \midrule
    1M & 148 & 6    & 35  & 255       & 747       & 764       & 728\\
    2M & 148 & 10   & 67  & 453       & \(\star\) & 770       & 760\\
    4M & 148 & 18   & 131 & 722       & \(\star\) & 775       & 774\\
    8M & 148 & 34   & 255 & 710       & $\star$   & 773       & 772\\
   16M & 148 & 65   & 483 & 732       & $\star$   & \(\diamond\) & 768\\
   32M & 148 & 127  & 781 & \(\star\) & $\star$   & \(\diamond\) & 763\\
  \end{tabular}

  \begin{tabular}{r >{\raggedleft\arraybackslash}p{2em}rrrrrr}
    \toprule
    \multirow{2}{4em}{\centering No.~of strings}
    & \multicolumn{7}{c}{Dataset: random} \\
    \cmidrule{2-8}
     & $l$ & \bwtlcpem & \beetlem & \beetl & \gsais & \egsa & \egap\\
    \midrule
    1M & 25  & 6   & 35  & 214   & 761      & 761       & 724 \\
    2M & 28  & 10  & 67  & 392   & \(\star\)& 767       & 766 \\
    4M & 27  & 18  & 131 & 734   & \(\star\)& 775       & 781 \\
    8M & 28  & 33  & 255 & 726   & $\star$  & 771       & 772 \\
   16M & 30  & 65  & 483 & 729   & $\star$  & \(\diamond\) & 762 \\
   32M & 32  & 128 & 777 &\(\star\)&$\star$ & \(\diamond\) & 760 \\
    \bottomrule
  \end{tabular} \caption{Peak RAM usage on the NA24385 and random datasets using a PC with 1GB RAM.
    The first column indicates the number of sequences in the dataset.
    Column $l$ indicates the maximum value of the LCP array on that dataset.
    \(\star\) means that the tool required more than 1GB of RAM.
    \(\diamond\) means that the tool crashed because of disk space exhaustion.
  }
  \label{table:experimental-evaluation-memory-1g}
\end{table}

Table~\ref{table:experimental-evaluation-memory-1g}, reports the RAM usage (in
Megabytes) required to compute the BWT and the LCP array.
Just as for Table~\ref{table:experimental-evaluation-times-1g}, the first column
indicates the number of sequences in the datasets, whereas column $l$ indicates
the maximum value of the LCP array on that dataset. As before, symbols \(\star\)
and \(\diamond\) mean that the tool could not complete the execution in our
environment since it exhausted the RAM or the available disk space,
respectively.
Notice that \bwtlcpem is always the tool requiring the smallest amount of memory,
by a factor at least 6.

\section{Conclusions}
\label{sec:conclusions}

We have presented a new lightweight algorithm to compute the BWT and the LCP array of a
set of $m$ strings, each  $k$ characters long, based on applying a backward strategy for merging partial BWTs.
More precisely, our algorithm has an \(\mathcal{O}(mkl)\) time and I/O volume, and uses
\(\mathcal{O}(k + m)\) main memory to compute the BWT and LCP array, where $l$ is the maximum
value in the LCP array.
Our time complexity and I/O volume are in the worst case as those of the best previously available algorithms.
The experimental analysis shows that our approach is competitive with the best
available external-memory methods and that its advantage is noticeable
on large inputs when the available RAM is limited.

The approach presented here may be further investigated in other research directions, as for example in the case of arbitrary alphabets and for  collection of strings in other contexts, such as in dealing with dictionaries where the parameter  $l$ may be smaller than the size of the input strings.
Theoretically, it is of interest to investigate the open question whether the
optimal time \(\mathcal{O}(mk)\) time can be achieved for computing the BWT.
Some recent results (see, for example, \cite{kempa2020resolution}) may also
suggest that running times of BWT-related algorithms could be improved on
compressible inputs.
Investigating if these results have an implication to our algorithm is an
interesting research direction.

On the other end, the prototype called \bwtlcpem implementing our approach is still a proof of concept, although carefully developed, and
future releases of the tool could improve its performance by, for example, a better buffering strategy of the
input and output files, asynchronous I/O, or better representation (i.e., fast compression) of the
intermediate files.

\section*{Acknowledgements}
We would like to thank Giovanni Manzini for the discussion on the implementation and on the comparison with the software \egap.
This project has received funding from the European Union’s Horizon 2020 research and innovation programme under the Marie Skłodowska-Curie grant agreement No 872539.

\bibliography{abbreviations,bwt,haplotype}

\end{document}